\documentclass[reprint,prx,aps,superscriptaddress,twocolumn,showpacs,longbibliography,floatfix,nofootinbib]{revtex4-2}
\usepackage{graphicx,color}
\usepackage{amsmath,amsfonts,enumerate,amsthm,amssymb,bbm, physics}
\usepackage[colorlinks=true,citecolor=blue,linkcolor=magenta]{hyperref}
\usepackage{cleveref}
\usepackage{multirow}
\usepackage{bbold}
\usepackage{braket}
\usepackage{soul}
\usepackage[caption = false]{subfig}
\usepackage{scrextend}
\usepackage{xcolor}
\usepackage{dsfont}
\usepackage{nicefrac}
\usepackage[linesnumbered,ruled,vlined]{algorithm2e}
\usepackage{tabularx}
\usepackage{booktabs}
\usepackage{tikz}
\usepackage{qcircuit}
\usepackage{apptools}
\usepackage{mathtools}
\usepackage{balance}
\usepackage{appendix}
\usepackage{float}
\usepackage{bbm}
\usepackage{xfrac}
\usepackage{csquotes}

\newcommand{\ys}[1]{\textcolor{red}{#1}}

\xyoption{color}

\mathchardef\ordinarycolon\mathcode`\:
\mathcode`\:=\string"8000
\begingroup \catcode`\:=\active
  \gdef:{\mathrel{\mathop\ordinarycolon}}
\endgroup

\newtheorem{repinner}{Theorem}
\newenvironment{reptheorem}[1]{%
  \renewcommand\therepinner{#1}%
  \begin{repinner}
}{\end{repinner}}

\theoremstyle{plain}
\newtheorem{thm}{Theorem}
\newtheorem{corol}{Corollary}
\newtheorem{lem}{Lemma}
\newtheorem{propos}{Proposition}

\theoremstyle{definition}

\theoremstyle{remark}
\newtheorem{rmk}{Remark}

\AtAppendix{\counterwithin{thm}{section}}
\AtAppendix{\counterwithin{corol}{section}}
\AtAppendix{\counterwithin{lem}{section}}
\AtAppendix{\counterwithin{propos}{section}}
\AtAppendix{\counterwithin{defn}{section}}
\AtAppendix{\counterwithin{rmk}{section}}
\AtAppendix{\counterwithin{ex}{section}}

\def\<{\langle}

\def\x{\boldsymbol{x}}

\def\>{\rangle}
\def\<{\langle}

\renewcommand{\ket}[1]{|#1\rangle}               
\renewcommand{\bra}[1]{\langle #1|}              

\usepackage{mathpazo}

\makeatletter
\DeclareRobustCommand{\rvdots}{%
  \vbox{
    \baselineskip4\p@\lineskiplimit\z@
    \kern-\p@
    \hbox{.}\hbox{.}\hbox{.}
  }}
\makeatother

\DeclareFontFamily{U}{wncy}{}
    \DeclareFontShape{U}{wncy}{m}{n}{<->wncyr10}{}
    \DeclareSymbolFont{mcy}{U}{wncy}{m}{n}
    \DeclareMathSymbol{\Sh}{\mathord}{mcy}{"58} 
\begin{document}
\allowdisplaybreaks

\title{Optimal Coherent Quantum Phase Estimation via Tapering}

\affiliation{Department of Computer Science, Cornell University, Ithaca, New York, 14850, USA.}
\affiliation{Joint Center for Quantum Information and Computer Science,
NIST/University of Maryland, College Park, MD, USA.}
\affiliation{Department of Computer Science, University of Maryland, College Park, MD, USA.}
\affiliation{Information Sciences, Los Alamos National Laboratory, Los Alamos, NM, USA.}

\author{Dhrumil Patel}
\thanks{The first two authors contributed equally to this work.}
\affiliation{Department of Computer Science, Cornell University, Ithaca, New York, 14850, USA.}
\affiliation{Information Sciences, Los Alamos National Laboratory, Los Alamos, NM, USA.}

\author{Shi Jie Samuel Tan} 
\thanks{The first two authors contributed equally to this work.}
\affiliation{Joint Center for Quantum Information and Computer Science,
NIST/University of Maryland, College Park, MD, USA.}
\affiliation{Department of Computer Science, University of Maryland, College Park, MD, USA.}
\affiliation{Information Sciences, Los Alamos National Laboratory, Los Alamos, NM, USA.}

\author{Yi\u{g}it Suba\c{s}\i} 
\affiliation{Information Sciences, Los Alamos National Laboratory, Los Alamos, NM, USA.}

\author{Andrew T. Sornborger} 
\affiliation{Information Sciences, Los Alamos National Laboratory, Los Alamos, NM, USA.}

\date{\today}

\begin{abstract}
\noindent
Quantum phase estimation is one of the fundamental primitives that underpins many quantum algorithms, including Shor's algorithm for efficiently factoring large numbers. Due to its significance as a subroutine, in this work, we consider the coherent version of the phase estimation problem, where given an arbitrary input state and black-box access to unitaries $U$ and controlled-$U$, the goal is to estimate the phases of $U$ in superposition. Most existing phase estimation algorithms involve intermediary measurements that disrupt coherence. Only a couple of algorithms, including the standard quantum phase estimation algorithm, consider this coherent setting. However, the standard algorithm only succeeds with a constant probability. To boost this success probability, one can employ the coherent median technique, resulting in an algorithm with asymptotically optimal query complexity (the total number of calls to $U$ and controlled-$U$). However, this coherent median technique requires a large number of ancilla qubits and a computationally expensive quantum sorting network.

To address this, in this work, we propose an improved version of the standard algorithm called the tapered quantum phase estimation (tQPE) algorithm, which leverages tapering (or window) functions commonly used in classical signal processing. Our algorithm achieves the asymptotically optimal query complexity without requiring the expensive coherent median technique to boost success probability. Moreover, we find the absolutely optimal taper — not only in the asymptotic scaling but in terms of exact performance. We provide an efficiently preparable ancilla state based on an approximation of the optimal taper, which incurs at most a factor-of-two increase in the probability of error, thereby maintaining near-optimal performance in practice. 
In the appendices, we give an explicit construction of the taper state preparation circuit. Finally, we derive an error bound for coherent QPE when the phase estimate is used as a control and subsequently uncomputed. 
\end{abstract}

\maketitle

\tableofcontents

\section{Introduction}
Quantum phase estimation (QPE) has been central to the field of quantum computing since its introduction~\cite{shor1994algorithms, kitaev1995quantum}.
It has been used in Shor's algorithm for efficiently factoring large numbers~\cite{shor1994algorithms}, in the Harrow–Hassidim–Lloyd (HHL) algorithm to solve a system of linear equations~\cite{harrow2009quantum}, for quantum amplitude estimation~\cite{brassard2002quantum}, for quantum principal component analysis~\cite{lloyd2014quantum}, for fast Quantum Merlin-Arthur (QMA) amplification~\cite{nagaj2009fast}, and as a subroutine in many other applications ~\cite{wiebe2012quantum,lloyd2013quantum,chowdhury2018improved,wright2021automatic,engel2021linear,an2021quantum, rall2020quantum}.

At its core, the goal of QPE is to estimate the phase of an eigenvalue of a given unitary. Let $U$ be a unitary acting on a $d$-dimensional Hilbert space, $\mathcal{H}$, and suppose that $e^{2\pi i\theta}$ is one of the eigenvalues of $U$, where $\theta$ lies in the range $[0, 1)$.
Then, the QPE problem is to output an estimate $\tilde{\theta}$ that is $\delta$-close to $\theta$ with a probability at least $1-\epsilon$ for some $\delta, \epsilon >0$.

To solve this problem, it is common to assume that we are given black-box access to $U$ and its controlled version (controlled-$U$), as well as sample access to the eigenvector $|\psi_{\theta}\rangle$ of $U$ corresponding to the eigenvalue $e^{2\pi i \theta}$. For our purposes, we define the cost of an algorithm in terms of its query complexity, which is the number of times $U$ and controlled-$U$ are applied. An upper bound on the gate complexity can be obtained from the query complexity for a given circuit implementation of~$U$ by simple multiplication.

Probably the simplest approach to the QPE problem is the Hadamard test. It estimates the real and imaginary parts of the overlap $\langle \psi_{\theta}|U|\psi_{\theta}\rangle$, giving $e^{2\pi i \theta}$. This can then be used to estimate $\theta$. 
However, this approach requires $O\!\left(\delta^{-2}\log(\sfrac{1}{\epsilon})\right)$ queries to achieve a desired precision, $\delta$.
A quadratic improvement in precision, i.e., $O\!\left(\delta^{-1}\log(\sfrac{1}{\epsilon})\right)$, can be achieved using the well-known algorithm proposed by Kitaev in 1995~\cite{kitaev1995quantum}.
A crucial insight for this quadratic speedup was to extract the phase value bit-by-bit using controlled-$U^{2^j}$, rather than simply using controlled-$U$ in each iteration like in the Hadamard test. Interestingly, the authors of~\cite{mande2023tight} recently showed that any algorithm solving the QPE problem requires $\Omega\!\left(\delta^{-1} \log (\sfrac{1}{\epsilon})\right)$ queries, making Kitaev's algorithm optimal in that sense.

In many practical applications of QPE, it is not realistic to assume that we have access to the exact eigenvector, $|\psi_\theta\rangle$. Instead, we have access to an arbitrary quantum state, $|\psi\rangle$, that can be written as a superposition of eigenvectors $\left \{|\psi_{\theta_r}\rangle\right\}_{r=0}^{d-1}$ of $U$:
\begin{align}
    |\psi\rangle =  \sum_r c_r |\psi_{\theta_r}\rangle.
\end{align}
Suppose that the eigenvalues corresponding to these eigenvectors are $\left \{e^{2\pi i\theta_r}\right\}_{r=0}^{d-1}$. A more practical goal is then to estimate phases coherently in superposition, i.e., prepare a state
\begin{align}
    \sum_r c_r \left(\ket{\omega_{r}^{\delta}} + \ket{\omega_{r}^{\delta \perp}}\right)  \ket{\psi_{\theta_r}},\label{eq:coh-state}
\end{align}
where
\begin{align}
\ket{\omega_{r}^{\delta}} = \sum_{s:|\theta_r - \tilde{\theta}_s|\le \delta}d_{r,s}\ket{\tilde{\theta}_{s}}
\end{align}
is the part of the ancilla state that encodes a $\delta$-approximation to $\theta_r$ in the computational basis. We want $\|\ket{\omega_{r}^{\delta}}\| \ge \sqrt{1-\epsilon}$ which implies that a $\delta$-approximation can be read out with probability $\ge 1-\epsilon$. 
In practice, the error probability will depend on the true phase: $\|\ket{\omega_{r}^{\delta}}\| = \sqrt{1-\epsilon(\theta)}$, and usually we demand $\max_\theta \epsilon(\theta) \le \epsilon$. 
This coherent setting is more useful because most algorithms employ QPE as a subroutine, rather than as a standalone algorithm.

Both the Hadamard test and Kitaev's algorithm fall under the category of `iterative' QPE algorithms. An iterative algorithm consists of multiple successive computations, each consisting of simple quantum circuits. Each iteration, however, consists of measurements at the end of the circuit evaluation followed by some classical post-processing. Applying these algorithms to an arbitrary state aside from an eigenvector of $U$ causes decoherence due to these inherent measurements in each iteration. Ultimately, this leads to the preparation of an incoherent state as opposed to the pure coherent state, given by~\eqref{eq:coh-state}.

To the best of our knowledge, only two algorithms exist in the literature that can perform phase estimation coherently~\cite{cleve1998quantum, Rall2021}, with the standard QPE algorithm (also known as the textbook QPE algorithm) being the most well-known~\cite{cleve1998quantum}. This algorithm incorporates the (inverse) quantum Fourier transform (QFT) as a subroutine and uses $O(\delta^{-1})$ queries to prepare the desired coherent state with probability at least $\sfrac{8}{\pi^2}$ in a single run. However, this constant success probability may not be sufficient for quantum algorithms that use QPE as a subroutine. Therefore, boosting this success probability to the desired $1 - \epsilon$ becomes extremely crucial. To achieve this, the standard QPE algorithm is executed $O(\log(\sfrac{1}{\epsilon}))$ times in parallel, and the median of the outputs is computed coherently~\cite{nagaj2009fast}. Therefore, the overall algorithm, i.e., the standard QPE combined with the coherent median computation for boosting success probability, has query complexity of $O(\delta^{-1}\log(\sfrac{1}{\epsilon}))$. This query complexity is optimal, matching the corresponding lower bound $\Omega(\delta^{-1}\log(\sfrac{1}{\epsilon}))$. However, the coherent median step involves using a large number of ancilla qubits and a quantum sorting network~\cite{hoyer2002quantum, klauck2003quantum, beals2013efficient}, which is computationally expensive.

An alternative approach to boost the success probability, without using a sorting network, involves using $m=O(\log(\sfrac{1}{\epsilon}))$ additional ancilla qubits all prepared in the uniform superposition state, but this increases the overall query complexity to  $O\left( \delta^{-1} \epsilon^{-1}\right)$~\cite{cleve1998quantum}. This is because the query complexity of the standard QPE algorithm grows exponentially with the number of total ancilla qubits $p=\ell+m$, i.e., $O\left( \delta^{-1} 2^{m}\right)$. Here $\ell = O(\log(\sfrac{1}{\delta}) )$ is the minimum number of ancilla qubits required to achieve precision~$\delta$.  Substituting $m=O(\log(\sfrac{1}{\epsilon}))$ into this yields the stated complexity. It is worth noting that this complexity is exponentially worse in $\epsilon$ compared to using the coherent median approach. 

Therefore, in this work, we investigate the following question and answer it positively: can we improve the standard QPE algorithm such that it maintains the optimal query complexity of $O(\delta^{-1}\log(\sfrac{1}{\epsilon}))$ without employing the median approach and the associated computationally-expensive quantum sorting network?

\subsection{Contributions}
We propose an improved version of the standard QPE algorithm, which we call the \textit{tapered QPE (tQPE) algorithm}. The rationale behind this name will become apparent later in the discussion. We show that the tQPE algorithm uses exponentially fewer additional qubits, $m$, to achieve a success probability arbitrarily close to one. This improvement directly leads to an exponentially smaller query complexity in terms of~$\epsilon$.

To be more precise, we demonstrate that our tQPE algorithm requires only $m= \lceil\log_2 \log (1/\epsilon) \rceil$ additional qubits to achieve a success probability at least $1-\epsilon$ (see Sec.~\ref{sec:average-case}).
We obtain this improvement by framing the problem as an optimization over the choice of the ancilla qubit state that maximizes the success probability. This results in a more effective choice for the initial state of the ancilla register. In the case of the standard QPE algorithm, the ancilla register is initialized to the uniform superposition state. 
However, we propose to replace this state with the state corresponding to the most-frequency-concentrated discrete prolate spheroidal sequence (DPSS), which is a widely used \textit{window/tapering} function in the field of classical signal processing. This motivates the name of our algorithm.
We will provide a brief note on DPSS and on window/tapering functions in general later in Sec.~\ref{sec:brief-note-on-dpss}.

The key idea behind this substitution is that DPSS maximizes signal concentration within a given spectral band. This is particularly important for QPE, as it helps in maximizing the probability of obtaining phase estimates that are $\delta$-close to the true phase~$\theta$.
To this end, the query complexity of our algorithm is then $O\!\left(\delta^{-1} \log(\sfrac{1}{\epsilon}) \right)$, which scales exponentially better in terms of $\epsilon$ than that of the standard QPE algorithm, which, as mentioned before, is $O\!\left ( \delta^{-1}\epsilon^{-1}\right)$. Furthermore, the query complexity of our algorithm saturates the lower bound $\Omega\!\left(\delta^{-1} \log(\sfrac{1}{\epsilon})\right)$, which means that it is optimal. Having said that, we would also like to emphasize one of the implications of our result: tQPE can be used directly for fast QMA amplification instead of running the standard QPE algorithm $O(\log (\sfrac{1}{\epsilon}))$ times for computing the median. This suggests an alternative approach for fast QMA amplification that employs tapering functions.

A natural question following the aforementioned result is whether one can initialize the ancilla register to this DPSS state efficiently. We answer this question by providing explicit algorithms for doing so, which we describe in detail in Sec.~\ref{sec:taperprep} and App.~\ref{app:TaperPreparation}. We subsequently show that the gate complexity of this algorithm is comparable to that of the standard QPE algorithm  for initializing the uniform superposition state, up to $\log\log$ factors. 

The tQPE algorithm, like the standard QPE algorithm, involves applying the inverse QFT to the ancilla qubits. However, since tQPE uses exponentially fewer additional qubits, the inverse-QFT circuit is also smaller compared to that used in the standard algorithm. This can be seen as follows. In~\cite{Hales}, the authors showed that the gate complexity of the inverse QFT acting on $p$ qubits is $O(p\log p)$. Therefore, the gate complexity of the inverse QFT in the case of the standard algorithm is \[O\!\left(\left(\log\frac{1}{\delta} + \log\frac{1}{\epsilon}\right) \log\left(\log\frac{1}{\delta} + \log\frac{1}{\epsilon}\right) \right)\] 
because the number of ancilla qubits required for this algorithm is $p = O\!\left(\log(\sfrac{1}{\delta}) + \log(\sfrac{1}{\epsilon})\right)$.
On the other hand, the gate complexity of this transform in the case of tQPE is \[O\!\left(\left(\log\frac{1}{\delta} + \log\log\frac{1}{\epsilon}\right) \log\left(\log\frac{1}{\delta} + \log\log\frac{1}{\epsilon}\right) \right)\] 
because the ancilla qubits required for this algorithm is $p = O\!\left(\log(\sfrac{1}{\delta}) + \log \log(\sfrac{1}{\epsilon})\right)$, a significant improvement.

We also study a special scenario where one is not allowed to use any additional qubits, i.e., $m=0$ (see Sec.~\ref{sec:special_case_Delta}) and one wishes to maximize the success probability of outputting one of the two nearest estimates when $\theta$ lies exactly in between two phase estimates. The optimal input states for the ancilla qubits for this scenario define a two-dimensional subspace which includes a particular kind of sinusoidal sequence, defined in~\eqref{eqn:sine_taper}. Additionally, we perform numerics, plotting the success probability as a function of the distance between the true value and the closest phase estimate (see Sec.~\ref{sec:numerics}). We carry out these numerics for three different input states of the ancilla register: 1) the uniform superposition state, 2) the DPSS state, and 3) the above sinusoidal state. This analysis demonstrates that DPSS performs well over the entire range as compared to the uniform superposition state.

In Sec.~\ref{sec:uncomputing} we address the issue of coherent uncomputation of the phase estimate. In many applications such as HHL algorithm~\cite{harrow2009quantum} and Quantum Metropolis Sampling~\cite{temme2011quantum}, QPE is used to implement a unitary that is controlled by the phase estimate which is subsequently uncomputed. In Theorem~\ref{thm:uncomputing} (proven in App.~\ref{app:uncomputing}) we give an error bound on such uses of QPE that applies to any algorithm for coherent QPE. We anticipate that this bound will be useful for analyzing such algorithms as it accounts for multiple errors that in previous analyses had to be treated individually.

In Sec.~\ref{sec:randomization} we show that by randomizing over the input phases as in~\cite{Lu2023unbiased}, we can reduce any instance of the QPE problem to an average instance for which DPSS taper is optimal. This proves that the DPSS taper achieves the optimal performance among all tapers not only asymptotically but in an absolute sense. In Sec.~\ref{sec:taperprep} and App.~\ref{app:TaperPreparation}, we give an explicit description of how approximate taper states with near-optimal performance can be prepared efficiently on a quantum computer. We find that for most applications, $m\le 4$ additional ancilla qubits might be sufficient.

\subsection{Overview}
Let $p$ be the number of ancilla qubits in the QPE circuit, initialized to the all-zeros state. For simplicity, we assume that we have sample access to one of the eigenvectors, $|\psi_{\theta}\rangle$, of $U$. In tQPE, as shown in Fig.~\ref{fig:circuitdiagram}, there are three main steps: 
\begin{enumerate}
    \item Initialize the ancilla qubits to a state $|\phi\rangle$. For now, this state serves as a placeholder and will be replaced later with a specific state depending on the problem at hand.
    \item Apply unitaries controlled-$U^{2^j}$ to the joint state, $|\phi\rangle |\psi_{\theta}\rangle$.
    \item Apply inverse QFT to the ancilla qubits.
\end{enumerate}
\begin{figure}
\resizebox{0.85\linewidth}{!}{
\Qcircuit @C=1em @R=.7em {
\mbox{}&&\mbox{\hspace{-0.45cm}\eqref{eq:taper-tQPE}}&&&&&&&&\mbox{\hspace{0.5cm}\eqref{eq:taper-tQPE-c-unitary}}& \hspace{1.2cm} \mbox{\eqref{eq:taper-tQPE-final-state}}&\\
 & \multigate{4}{U_{\text{PREP}}} & \qw & \qw & \qw & \qw & \dots & & \qw &\ctrl{5} & \qw  & \multigate{4}{\text{QFT}^\dagger} & \qw \\
& \ghost{U_{\text{PREP}}} \ar@[blue]@{--}[]+<2.3em,2.0em>;[d]+<2.3em,-4.5em> & \qw & \qw & \qw & \qw & \dots & & \ctrl{4} & \qw & \qw\ar@[blue]@{--}[]+<0.7em,2.0em>;[d]+<0.7em,-4.5em> & \ghost{QFT}\ar@[blue]@{--}[]+<2.2em,2.0em>;[d]+<2.2em,-4.5em> & \qw  \\
& \ghost{U_{\text{PREP}}} & \qw & \rvdots & \rvdots &  & \dots & & &  &  & \ghost{QFT} & \rvdots \\
& \ghost{U_{\text{PREP}}} & \qw & \qw & \ctrl{2} & \qw & \dots & & \qw & \qw & \qw & \ghost{QFT} & \qw \\
& \ghost{U_{\text{PREP}}} & \qw & \ctrl{1} & \qw & \qw & \dots & & \qw & \qw & \qw & \ghost{QFT} & \qw \\
\lstick{\ket{\psi}} & \qw & \qw & \gate{U^{2^0}} & \gate{U^{2^1}} & \qw & \dots &  & \gate{U^{2^{p-2}}} & \gate{U^{2^{p-1}}} & \qw & \qw & \qw 
\inputgroupv{2}{6}{0.6em}{4em}{\ket{0}^{\otimes p}}
}
}
\caption{Tapered QPE quantum circuit. The system is initialized in the state $\ket{0}^{\otimes p} \ket{\psi}$, where $\ket{\psi}$ is an arbitrary input state. $U_{\text{PREP}}$ is the unitary that prepares the taper state $\ket{\phi}$ (see App.~\ref{app:TaperPreparation} for $U_{\text{PREP}}$ that prepares optimal tapers approximately). Dashed blue lines labeled by equation numbers denote the corresponding state at each step.}
\label{fig:circuitdiagram}
\end{figure}

Notably, the standard QPE algorithm is a special case of tQPE, where the ancilla qubits are initialized to the uniform superposition state, i.e.,
\begin{align}
    |\phi\rangle = \frac{1}{2^{p/2}}\sum_i |i\rangle.
\end{align}

The application of the above three steps leads to the following transformation:
\begin{align}
    |0^{\otimes p}\rangle |\psi_{\theta}\rangle \xrightarrow{\text{tQPE}} \sum_{k=0}^{2^p-1}\hat{\phi}\!\left(\theta - \frac{k}{2^p}\right)\ket{k}\ket{\psi_{\theta}},\label{eq:overview-final-state}
\end{align}
where $\hat{\phi}(\cdot)$ is the discrete-time Fourier transform, defined explicitly in~\eqref{eq:FTofTaper}, of $|\phi\rangle$. Fig.~\ref{fig:taper} shows how the probabilities with which phase estimates are output depend on the true phase $\theta$. 

As we can see from the above transformation, the final state of the ancilla register is a superposition of all possible values of phase estimates, $\sfrac{k}{2^p}$. The success probability of the algorithm is given by
\begin{align}
    \sum_{k:\left| \theta - \sfrac{k}{2^p}\right| \leq \delta} \left|\hat{\phi}\!\left(\theta - \frac{k}{2^p}\right)\right|^{2} = 1-\epsilon(\theta)\; ,\label{eq:overview-sp1}
\end{align}
which in general depends on $\theta$ as we will see later. 
Intuitively, we want to choose a state, $|\phi\rangle$, of the ancilla register whose support in the frequency domain is concentrated within the frequency band $[-\delta, \delta]$ and is at least $1 - \epsilon$, where, as mentioned before, $\epsilon$ is a uniform upper bound on the error probability across all $\theta$, i.e.,  $\max_\theta \epsilon(\theta) \le \epsilon$.
Fig.~\ref{fig:taper} shows an example of a taper whose support is concentrated around zero in its frequency domain. 
\\
\begin{figure}
\centering
\hspace{-1.57em}\includegraphics[width=0.505\textwidth]{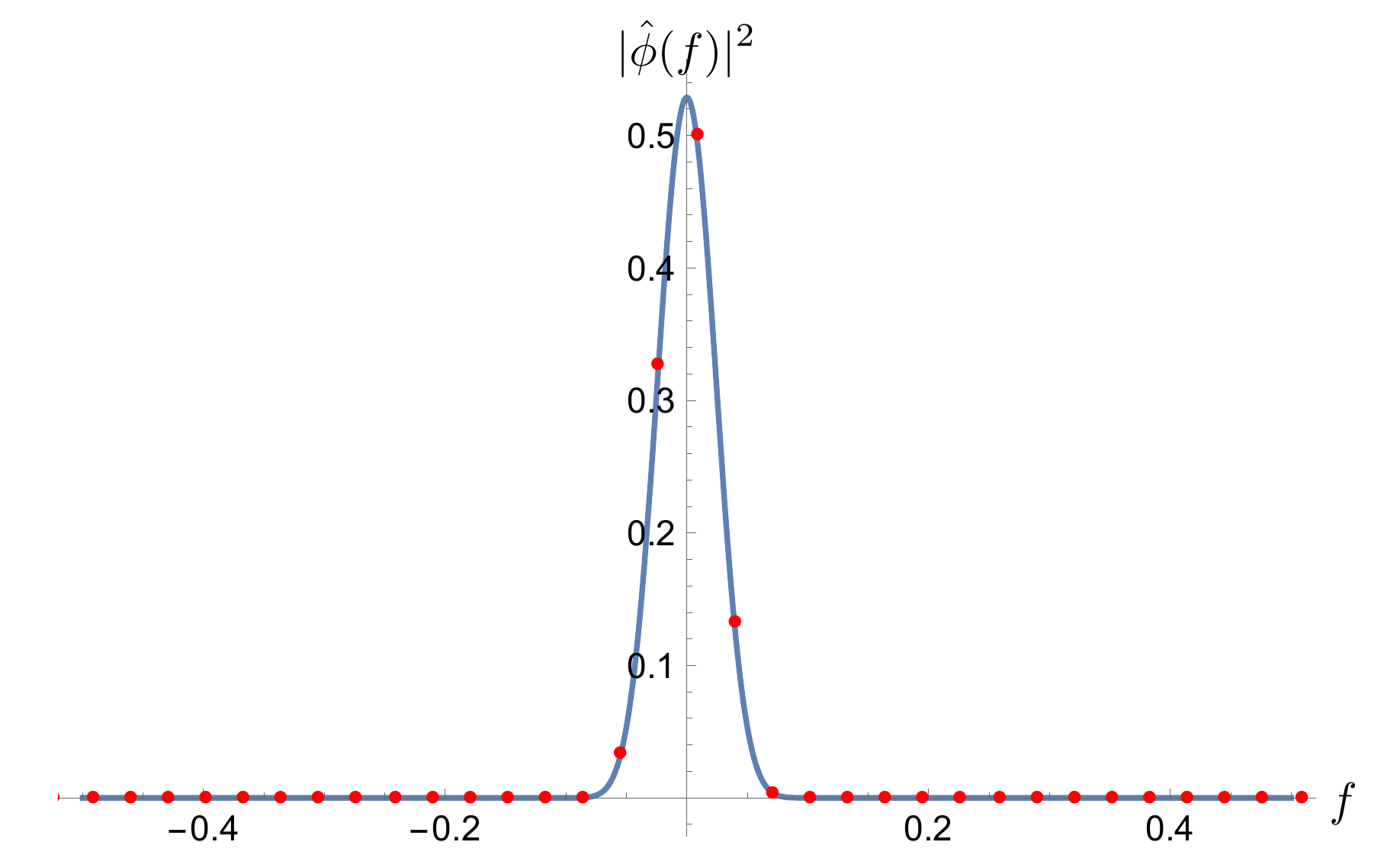}
    \caption{For the DPSS taper with $p = 5$, we plot the absolute value squared of the taper in the frequency domain (blue curve). The function evaluated at the discrete values $\theta-k/2^p$ (red dots) correspond to the squared amplitudes of the final ancilla state in~\eqref{eq:overview-final-state} for some choice of $\theta$. } 
    \label{fig:taper}
\end{figure}

We study error in two different settings: average- and worst-case. In the average-case error setting, we formulate an optimization problem that aims to find ancilla states, $|\phi\rangle$, that minimize the average error of the tQPE algorithm when the phases are sampled uniformly at random. 
We show that solving these optimization problems is equivalent to solving frequency concentration problems arising in classical signal processing, and the most-frequency-concentrated DPSS naturally emerges as the optimal solution in the average-case setting. We then show how worst-case instances of the QPE problem can be reduced to average-case instances by randomization of the phases, thus proving DPSS taper is optimal in the worst case as well. 

Before we discuss related works, we supply a summary table of parameters and their roles to help guide readers who are less familiar with QPE to navigate the dense notation. Readers can find an overview of our notation in Table~\ref{tab:notation}.

\begin{table}[t]
\centering
\renewcommand{\arraystretch}{1.15}
\setlength{\tabcolsep}{4pt}
\newcommand{\notedesc}[1]{\parbox[t]{0.50\columnwidth}{\raggedright #1}}
\begin{tabular}{lll}
\toprule
\textbf{Symbol} & \textbf{Domain} & \textbf{Interpretation} \\
\midrule
$\mathcal{H}$ & -
 & \notedesc{$d$-dimensional Hilbert space}\\
$U$ & $\mathsf{U}(\mathcal{H})$
 & \notedesc{Black-box unitary}\\
$\theta$ & $[0,1)$ & \notedesc{Eigenphase of $U$}\\
$|\psi_\theta\rangle$ & $\mathcal{H}$ & \notedesc{Eigenstate of $U$ with phase $\theta$} \\
$|\psi\rangle$ & $\mathcal{H}$ & \notedesc{Arbitrary input $\sum_r c_r |\psi_{\theta_r}\rangle$} \\
$\tilde{\theta}$ & $[0,1)$ & \notedesc{Phase estimate} \\
$\delta$ & $(0,1/2]$ & \notedesc{Target precision: success if $|\tilde\theta\!-\!\theta|\!\le\!\delta$} \\
$\epsilon(\theta)$ & $[0,1]$ & \notedesc{Error probability for true phase $\theta$}\\
$\epsilon$ & $(0,1]$ & \notedesc{Uniform upper bound on the error probability}\\
$\overline{\epsilon}^\text{opt}$ & $(0,1]$ & \notedesc{Optimal average error probability for the tQPE algorithm, attained by the DPSS taper} \\
$\overline{\epsilon}^\phi$ & $(0,1]$ & \notedesc{Average error probability for tQPE algorithm when $|\phi\rangle$ is used as a taper}\\
$\overline{\epsilon}$ & $(0,1]$ & \notedesc{Upper bound on the average error probability} \\
$\ell$ & $\mathbb{N}$ & \notedesc{Minimum number of ancilla qubits required to achieve precision $\delta$}\\
$m$ & $\mathbb{N}$ & \notedesc{Number of additional ancilla qubits used to boost success probability} \\
$p$ & $\mathbb{N}$ & \notedesc{Total number of ancilla qubits: $\ell+m$} \\
$N$ & $\mathbb{N}$ & \notedesc{Number of phase grid points: $2^{\ell+m}$} \\
$\mathcal{H}_{\text{t}}$ & - & \notedesc{$N$-dimensional Hilbert space corresponding to the ancilla/taper register}\\
$\Delta$ & $(-\tfrac{1}{2N},\tfrac{1}{2N}]$ & \notedesc{Offset between the true phase and the nearest phase estimate on the grid} \\
$K$ & $\{0,\dots,\frac{N}{2}-1\}$ & \notedesc{Used to define number of acceptable phase estimates: $2K{+}1$} \\
$W$ & $(0,0.5)$ & \notedesc{Half-bandwidth of the DPSS taper in the frequency domain: $W\!=\!\delta\!-\!1/\!(2N)\!$} \\
$\lambda_{\max}$ & $[0,1]$ & \notedesc{Maximum eigenvalue of DPSS kernel}\\
\bottomrule
\end{tabular}
\caption{Summary of notation.}
\label{tab:notation}
\end{table}

\subsection{Related work}

QPE is a well-researched problem and hitherto many quantum algorithms have been proposed to solve it under different settings. However, most of the quantum algorithms in the literature~\cite{kitaev1995quantum, svore2013faster, Ni2023, Lin2022, O’Brien2019, wan2022randomized, ChapeauBlondeau2020, Ding2023, gebhart2021bayesian, mande2023tight} are not suitable for the coherent setting that we are considering in this paper. This point is strongly emphasized in~\cite{Rall2021}. As such, only two algorithms, such as those presented in~\cite{cleve1998quantum} and~\cite{Rall2021}, provide coherent phase estimation algorithms and are thus suitable for use as subroutines in larger algorithms. The approach in~\cite{Rall2021} is complementary to ours and achieves similar query complexity. Their algorithm employs block-encoding techniques to obtain the phase estimate bit-by-bit and employs techniques from the quantum singular value transform (QSVT)~\cite{gilyen2019quantum, martyn2021grand} to boost the success probability.

Tapers are not new in the study of quantum algorithms. In 1999, Bu{\v{z}}ek \emph{et al.} used a sinusoidal taper for the construction of optimal quantum clocks~\cite{buvzek1999optimal}. The authors of~\cite{berry2000optimal} found that a particular kind of sinusoidal taper is optimal as an $N$-photon two-mode input state for obtaining an estimate
of the phase difference between two arms of an interferometer. The effect of a cosine taper on the quantum phase estimation has been studied in~\cite{rendon2022effects}, and the authors showed a cubic improvement of $m$ in terms of error probability $\epsilon$, i.e., $m=O(\log(1/\epsilon^{1/3}))$ as opposed to $m=O(\log(1/\epsilon))$ for the standard QPE algorithm. Additionally, the cosine taper has been employed in QPE-based algorithms, including the HHL algorithm for solving systems of linear equations~\cite{harrow2009quantum}. Bump functions have found application in spectral estimates based on time series analysis~\cite{somma2019quantum}. Tapers have also been explored within quantum spectral filtering methods aimed at efficient state initialization~\cite{fillion2017efficient}. Furthermore, approximate Gaussian tapers have been utilized in quantum algorithms for spectral density estimation~\cite{roggero2020spectral}, thermal state preparation~\cite{Moussa2019, chen2023quantum, chen2023efficient}, and ground-state energy estimation~\cite{Rendon2023}. 
The Kaiser taper~\cite{campbell1948fourier, kaiser1966digital, karnik2021improved, mcardle2022quantum} has been considered in QPE~\cite{berry2022quantifying} and achieves nearly asymptotically optimal scaling. We will further comment on Kaiser tapers and their relation to the DPSS taper in Sec.~\ref{sec:discussion}. Note that none of the tapers discussed above are optimally frequency-concentrated, resulting in side-lobes that warrant further reduction.

Previous works have used randomness to improve the performance of QPE. In 2023, Lu and Lin applied  phase randomization to achieve unbiased phase estimation~\cite{Lu2023unbiased}. By applying random single-qubit rotations to the ancilla register, a random offset to $\theta$ is applied which removes the bias that plagues QPE. We note that we essentially utilize their technique in our work for our worst-case to average-case reduction to argue for the optimality of our QPE approach. Linden and De Wolf also proved that randomization of the inputs to QFT can help achieve good worst-case performance for phase-estimation~\cite{linden2022average}. A different work by Wan, \emph{et al.} showed that randomizing how the unitary powers are sampled can help convert algorithmic errors into stochastic errors, leading to a doubly random approach which once again facilitates the reduction of the worst-case instances to the average-case instance~\cite{wan2022randomized}.

\section{A brief note on Discrete Prolate Spheroidal Sequences}\label{sec:brief-note-on-dpss}
In the field of classical signal processing and statistics, DPSS, sinusoidal sequences, and other such sequences are referred to as window functions, tapering functions, or simply tapers. 
These functions are widely used to analyze and modify the frequency spectrum of a given signal.
It is important to note that the uniform superposition state used in the standard QPE algorithm is also a type of window function known as a rectangular window. This function is also sometimes referred to as the {\it tophat} taper, and we will use this name throughout our paper to refer to this taper.

In a series of seminal papers in the field of signal analysis~\cite{slepian1961prolate, landau1961prolate, landau1962prolate, slepian1964prolate, slepian1978prolate}, Slepian, Pollack, and Landau studied the extent to which a time-limited signal can be band-limited. In other words, they investigated how much the Fourier transform of a signal can be concentrated in a small interval in the frequency domain (i.e., band-limited), given that the signal is only non-zero in a finite interval in the time domain (i.e., time-limited). 
The discrete-time case was studied in the fifth paper of this series of papers~\cite{slepian1978prolate}. In this paper, Slepian introduced DPSSs (also called Slepian sequences) as eigenvectors of a kernel arising from a particular frequency concentration algorithm and demonstrated that there exists a DPSS that is both time-limited and maximally band-limited. 

For completeness, we provide a detailed derivation of this in App.~\ref{app:classicaldpss}. In brief, Slepian constructed and optimized a cost function to find a discrete sequence, whose Fourier transform is maximally concentrated in a given bandwidth. The optimization results in a linear kernel, whose leading eigenvector solves the problem. Furthermore, the authors of~\cite{percival1993spectral} showed how to efficiently compute the entries of this eigenvector numerically.

Due to the ability of DPSSs to concentrate the support of the Fourier-transformed function in a small interval in the frequency domain, while also maintaining finite support in the time domain, it is well-suited for a wide variety of applications in signal processing, including signal filtering, and high-resolution spectral and harmonic analysis methods.

Classical signal analysis methods based on Slepian \emph{et al.}'s analysis have resulted in significant advances in spectral and harmonic analysis~\cite{slepian1978prolate, thomson1982spectrum, park1987multitaper, mitra1999analysis, haykin2009spectrum, sornborger2012multivariate}. Maximizing signal concentration within a given spectral band is not only desirable for classical signal analysis but is also of special significance for QPE as we will see in this paper. 
Intuitively, the band-limiting property of a taper is important especially for QPE in the sense that it helps in increasing the probability of outputting phase estimates that are $\delta$-close to the true phase~$\theta$.

\section{Tapered Quantum Phase Estimation}
\label{sec:tQPE}

For simplicity, we focus on a particular eigenvector $|\psi_{\theta}\rangle$ of $U$ with eigenvalue $e^{2\pi i\theta}$ and lay out the details of our tQPE algorithm. This approach can be readily extended to a superposition of eigenvectors as well. To begin with, as shown in Fig.~\ref{fig:circuitdiagram}, we prepare an ancilla register in a state $\ket{\phi}$ defined as 
\begin{equation}
    \ket{\phi} \coloneqq \sum_{n = 0}^{N-1}\phi[n]\ket{n},\label{eq:taper-tQPE}
\end{equation}
where $\phi[n]\in \mathbb{C}$ for all $n \in \{0, \ldots, N-1\}$. Also, $N \coloneqq 2^p$, where $p$ is the total number of qubits in the ancilla register. Throughout this paper, we refer to the state of the ancilla register, $\ket{\phi}$, as a taper, and we use the terms ``taper register" and ``ancilla register" interchangeably. For the time being, please note that $\ket{\phi}$ is an arbitrary state, with a particular choice specified later. We then apply the conditional unitary
\begin{equation}
    \label{eq:selectU}
    \sum_{n=0}^{N-1}\ketbra{n}{n}\otimes U^n
\end{equation}
to $\ket{\phi}\ket{\psi_{\theta}}$, resulting in the state
\begin{equation}
    \sum_{n=0}^{N-1}\phi[n]e^{2\pi i \theta n}\ket{n}\ket{\psi_{\theta}}. \label{eq:taper-tQPE-c-unitary}
\end{equation}

\noindent
Subsequently, we apply the inverse QFT (QFT$^{-1}=$ QFT$^{\dagger}$) on the ancilla register. This transforms the basis $\left\{\ket{n}\right\}_{n=0}^{N-1}$ as
\begin{equation}
    \ket{n} \to \frac{1}{\sqrt{N}}\sum_{k=0}^{N-1}e^{-2\pi i nk/N}\ket{k},
\end{equation}
giving us the following final state:
\begin{equation}
    \sum_{k = 0}^{N-1}\left(\frac{1}{\sqrt{N}}\sum_{n=0}^{N-1}\phi[n]e^{2\pi i n\left(\theta - k/N\right)}\right)\ket{k}\ket{\psi_{\theta}}.\label{eq:dft_intro}
\end{equation}

Note that the expression above inside the parentheses is the discrete-time Fourier transform, defined as
\begin{equation}
    \hat{\phi}\!\left(f\right) \coloneqq \frac{1}{\sqrt{N}}\sum_{n=0}^{N-1}\phi[n]e^{2\pi i nf}, \label{eq:FTofTaper}
\end{equation}
of the time-limited signal, $\phi[n]$, evaluated at the frequency, $\theta - k/N$. Using this, the final state of the algorithm, given by \eqref{eq:dft_intro}, can be expressed more concisely as: 
\begin{equation}
    \sum_{k=0}^{N-1}\hat{\phi}\!\left(\theta - k/N\right)\ket{k}\ket{\psi_{\theta}}, \label{eq:taper-tQPE-final-state}
\end{equation}
where we can think of $\left \{\hat{\phi}\!\left(\theta - k/N\right)\right \}_{k}$ as the probability amplitudes of the discrete frequencies $\left \{ k/N\right \}_{k}$ indexed by $k$. In other words, if the first register corresponding to the taper is measured, the tQPE algorithm outputs a phase estimate, $k/N$, with probability $\left|\hat{\phi}\left(\theta - k/N\right)\right|^2$ \ys{(see Fig.~\ref{fig:taper})}.

In App.~\ref{app:Convolution}, we show that the coefficient in (\ref{eq:taper-tQPE-final-state}) may be understood as the convolution
\begin{equation}
 \hat{\phi}\!\left(\theta - f\right) = \left( \hat{\phi} * \delta  \right)(\theta-f)
\end{equation}
evaluated at $f = k/N$, where $*$ denotes convolution and $\delta(\cdot)$ is the Dirac delta function (not to be confused with the precision parameter $\delta$ used elsewhere).
From the above equality, observe that the taper is centered about, and independent of, the frequency $\theta$. Conceptually, and important for our discussion below, this distinguishes the taper as a variable function that can take various forms depending on the particular optimization problem of interest.

With that in place, recall that the goal of tQPE is to prepare the final state, given by~\eqref{eq:taper-tQPE-final-state}, such that the following holds:
\begin{align}
    \label{eq:epstheta}
    \epsilon(\theta):=\sum_{k:\left| \theta - k/N\right| >\delta} \left|\hat{\phi}\left(\theta - k/N\right)\right|^{2} \leq \epsilon.
\end{align}
To simplify the analysis, we choose $\delta = 2^{-\ell -1}$ for some positive integer $\ell$ without loss of generality. 
If $\theta$ can be expressed exactly on $\ell$ bits, then there exists a $k \in \{0, 1, \ldots, N-1\}$ such that $\theta = k/N$. 
Otherwise, there is always a phase estimate $k'/N$ such that $\left|\theta - k'/N\right| \le 2^{-p-1}$, and a $\theta$ that saturates this bound (i.e., it lies exactly between two possible phase estimates). 
This constraint naturally arises because the resolution of the ancilla register is $2^{-p}$. Therefore, we only consider the case where $\delta \geq 2^{-p-1}$ which implies that $p \geq \ell$. Thus, we set $p = \ell + m$ for some $m \geq 0$ where $m$ can be interpreted as the number of additional qubits used to boost the success probability of estimating phases that are $\delta$-close to~$\theta$.

As we can see, the standard QPE algorithm is a special case of the tQPE algorithm, where the taper being used is the tophat taper ($\phi[n] = 1/\sqrt{N}; \forall $n$ $). 
For $m = 0$ (i.e., no additional qubits used to boost the success probability), the tQPE algorithm using the tophat taper outputs a $\delta$-close phase estimate, i.e., one of the two grid points closest to the true phase, with probability at least $8/\pi^2 \approx 0.81$.
As previously stated in the Introduction section, in order to increase the success probability further, a natural approach is to increase $m$. 
By doing so, the number of $\delta$-close phase estimates increases, which in turn increases the total probability of outputting such estimates. 
Specifically, the authors of~\cite{cleve1998quantum} showed that it is sufficient to choose $m = \lceil \log_2\left(1/(2\epsilon) + 1/2\right)\rceil$ to boost the success probability to at least $1 - \epsilon$ when considering the tophat taper.

Although increasing $m$ can boost the success probability, it also increases the computational cost of the algorithm, since the number of queries to $U$ is proportional to $2^{p} = 2^{\ell+m}$ (see  Fig.~\ref{fig:circuitdiagram}).
It also increases the size of the inverse-QFT circuit because, as mentioned before, the size of this circuit is of the order $O(p\log p)$.
Thus, it is of utmost importance to choose the taper, $\ket{\phi}$, optimally so as to minimize $m$ for a given~$\epsilon$.

\section{Optimal Tapers}
\label{sec:optimaltapers}

In what follows, we derive optimal tapers for two settings: first, when the relative position of the true phase with respect to the grid of phase estimates is known; and second, for the average-case scenario, which we define later. We then analyze the worst-case error behavior of these tapers.

\subsection{The Optimization Problem}

We now state the optimization problem more formally:
\begin{align}
    \max_{|\phi\rangle \in \mathcal{H}_{\text{t}}} \quad&  \sum_{k:\,\left|\theta - \frac{k}{N}\right|\leq \delta}\left|\hat{\phi}\left(\theta - \frac{k}{N}\right)\right|^2 \notag\\
    \text{subject to} \quad &  \sum_{n=0}^{N-1}\left|\phi[n]\right|^2 = 1.\label{eq:ideal-opt-prob}
\end{align}
Here, $\mathcal{H}_{\text{t}}$ denotes the $N$-dimensional Hilbert space corresponding to the taper register. The objective function above represents the success probability of the tQPE algorithm, i.e., probability of obtaining a phase estimate that is $\delta$-close to the exact value $\theta$, given $|\phi\rangle$ as the taper state. The constraint arises from the fact that $|\phi\rangle$ is a quantum state, and therefore, it must be normalized. 

Due to our choice of $\delta=2^{-\ell-1}$ and $p=\ell+m$, there are generically $2^m$ estimates that are $\delta$-close to $\theta$. In the special cases when the phase is on a grid point or exactly between two grid points, there are $2^m+1$ and $2^m + \delta_{m0}$ estimates that are $\delta$-close to $\theta$, respectively. Here, $\delta_{ij}$ denotes the Kronecker delta, which equals 1 when $i=j$ and 0 otherwise, and, again, should not be confused with the precision parameter $\delta$ used elsewhere.

Depending on the exact value of $\theta$, the index $k$ in the objective function of \eqref{eq:ideal-opt-prob} can range over different values. This significantly complicates the optimization problem. 
In order to simplify the above expression further, we define the following:
\begin{equation}
    \Delta \coloneqq \theta - \frac{k^*}{N},\label{eq:def-Delta}
\end{equation}
where $k^*/N$ is the phase estimate that is closest to $\theta$ on the grid. In other words, $\Delta$ is the difference between the best possible phase estimate afforded by the $p$-qubit taper and the true value, $\theta$. We can now rewrite $\hat{\phi}(\theta-k/N)$ as $\hat{\phi}(\Delta+j/N)$ since the sum is over a dummy index.

We observe that phase estimates corresponding to $|k|\le K$ with $K= 2^{m-1}-1$ are always $\delta$-close for $m\ge 2$ and $K=0$ for $m=0$. This generically leaves out one phase estimate that is furthest away from the true value. In the special cases when the phase is on a grid point or exactly between two grid points, the miscount is  $2- 2\delta_{m0}$ and $1$, respectively.  
Finally, all tapers of practical interest that we know of have Fourier transforms that are symmetric, have a peak at 0, and decay rapidly, meaning their value at $\Delta \pm(K+1)/N$ will be extremely small for large $K$ or equivalently small $\epsilon$. In light of this, we ignore the contribution of those points to the success probability in the optimization problem. 

With these considerations, we restate the optimization problem~\eqref{eq:ideal-opt-prob} as:
\begin{align}
    \max_{|\phi\rangle \in \mathcal{H}_{\text{t}}} \quad&  \sum_{j=-K}^K\left|\hat{\phi}\left(\Delta + \frac{j}{N}\right)\right|^2 \notag\\
    \text{subject to} \quad &  \sum_{n=0}^{N-1}\left|\phi[n]\right|^2 = 1.\label{eq:org-opt-prob}
\end{align}
To convert the above constrained optimization problem to an unconstrained one, we use the Lagrangian formulation:
\begin{multline}
      \mathcal{L}\left(|\phi\rangle, \lambda\right)
    =  \sum_{j = -K}^{K}\left|\hat{\phi}\left(\Delta +  \frac{j}{N}\right)\right|^2\\
    - \lambda\left(\sum_{n=0}^{N-1}\left|\phi[n]\right|^2 - 1\right).\label{eq:mainLagrange}
\end{multline}
where $\mathcal{L}$ is the Lagrangian and $\lambda \in \mathbb{R}$ is a Lagrange multiplier.

\subsection{Ideal Case Optimal Tapers}

Finding the optimal taper is equivalent to finding the stationary point of $\mathcal{L}$ that maximizes the objective function of \eqref{eq:org-opt-prob}. The stationary points of $\mathcal{L}$ can be found by setting all the partial derivatives of $\mathcal{L}$ to zero. Doing so, we get the following two conditions (see App.~\ref{app:QPDPSSDerivation} for a detailed derivation):
\begin{gather}
    \frac{1}{N}\sum_{n=0}^{N-1} e^{2\pi i \Delta(n - m)}\left(\frac{\sin\left(\pi(m - n)(2K+1)/N\right)}{\sin\left(\pi(m-n)/N\right)}\right)\phi[n]\nonumber \\
    \hspace{2.3cm}=\lambda\phi[m];\quad  \forall m \in\{0, \ldots, N-1\}, \label{eq:eigenvalue_full} \;\\
    \sum_{n=0}^{N-1}\left|\phi[n]\right|^2 = 1.
\end{gather}
The first equation is an eigenvalue equation, while the second equation is the normalization constraint. Note that all the eigenpairs $(|\phi\rangle, \lambda)$ of the above eigenvalue equation satisfy both the aforementioned conditions, and therefore, these eigenpairs are the stationary points of $\mathcal{L}$. Now, by substituting all these stationary points $(|\phi\rangle, \lambda)$ into the objective function of \eqref{eq:org-opt-prob}, we obtain the following:
\begin{align}
    \sum_{j = -K}^{K}\left| \hat{\phi}\left(\Delta + \frac{j}{N}\right)\right|^2 =  \lambda.
\end{align}
This implies that the stationary point  $(|\phi\rangle, \lambda)$ that maximizes the objective function is the eigenvector $|\phi\rangle$ with maximum eigenvalue, $\lambda$. Since the objective function is the probability of outputting one of the $2K+1$ phase estimates $\delta$-close to $\theta$, using the eigenvector of~\eqref{eq:eigenvalue_full} with the maximum eigenvalue as our taper will maximize this probability.

In App.~\ref{app:relationship}, we establish an explicit connection between the eigenvectors of   \eqref{eq:eigenvalue_full} and the periodic
discrete prolate spheroidal sequences (P-DPSS)~\cite{zhu2017eigenvalue}. We call the former quantum periodic discrete prolate spheroidal sequences (QP-DPSS), and it is easy to see that they depend on the value of $\Delta$. Then, combining results from App.~\ref{app:relationship} with known results for the P-DPSS in~\cite{xu1984periodic}, we observe that the eigenvector of \eqref{eq:eigenvalue_full} with the largest eigenvalue has an eigenvalue of $1$, regardless of the value of $\Delta$. In other words, there exists a taper for which the tQPE algorithm outputs one of the $\delta$-close $2K+1$ phase estimates with probability 1. 

It should not be surprising that the above observation holds true. To understand this intuitively, we can break it down into two cases. In the first case, we consider $\Delta = 0$, that is $\theta$ lies exactly on one of the grid points (see \eqref{eq:def-Delta}). In this case, the standard QPE algorithm using the tophat taper always returns $\theta$ with probability 1 because now $\theta$ itself is one of the possible phase estimates. 
On the other hand, if we consider the case where $0 < \Delta \leq \frac{1}{2N}$ or $-\frac{1}{2N} \leq \Delta < 0$, then
this case can also be converted into the $\Delta=0$ case by shifting the grid of possible phase estimates by $\Delta$. To accomplish this, we simply apply a unitary operator parameterized by $\Delta$ to the tophat taper. This operation shifts all the $2^p$ grid points by $\Delta$, such that $\theta$ now lies exactly on a grid point. After that, we output $\theta$ with probability 1 because it is now a possible phase estimate.

For example, for $K=0$ one can achieve the success probability 1 with the optimal taper given by
\begin{align}
    \phi^\Delta[n] = \frac{e^{2\pi i \Delta n}}{\sqrt{N}} \; .\label{eq:K0Delta}
\end{align}
As a check, we note that for $\Delta=0$, we recover the tophat taper, which is known to have success probability 1 when the true phase is on the grid of estimates. In Fig.~\ref{fig:sinevscosine} we plot the success probability of this taper for $\Delta=\pm1/2N$. The above procedure is described in greater detail in App.~\ref{app:relationship}.

\subsubsection{Special Case: $\Delta=\pm 1/2N$}
\label{sec:special_case_Delta}
In this case, there are pairs of grid points that are equidistant to the exact solution. So, it makes sense to consider the probability of outputting an even number of estimates as opposed to an odd number. More specifically, we will focus on the case where we only consider the nearest two estimates. Eq.~\eqref{eq:K0Delta} shows tapers that return the closest estimate with probability one for arbitrary $\Delta$. In the special case $\Delta=\pm 1/2N$, there are actually two such tapers: one shown in Eq.~\eqref{eq:K0Delta} and the other obtained by letting $\Delta \rightarrow - \Delta$. One of these tapers outputs the larger and the other the smaller of the two closest estimates with unit probability. Any linear combination of these tapers also succeeds with unit probability. One might be interested in combinations that output each closest estimate with equal probability. One such pair of tapers is
\begin{align}
    \phi^\text{sin}[n] &= \frac{\sin(\pi n/N)}{\sqrt{N/2}} \; ,\label{eqn:sine_taper} \\
    \phi^\text{cos}[n] &= \frac{\cos(\pi n/N)}{\sqrt{N/2}} \; .
\end{align}
Although these two tapers perform the same on $\Delta=\pm 1/2N$, their performance on other $\Delta$ is very different. In Fig.~\ref{fig:sinevscosine} we show that the sine taper is superior to the cosine taper because it performs better at all other values of $\Delta$. 
Note that $\phi^{\pm1/2N}$ has similar properties to the tophat taper by construction but the sine taper behaves qualitatively differently. This is because in constructing the sine taper we used the additional degree of freedom afforded by taking a linear combination of $\phi^{\pm1/2N}$ to improve the performance of the taper over all $\Delta$. 
In fact, the sine (cosine) taper can be obtained by minimizing (maximizing) the average-case error defined in the next section among all linear combinations between the two QP-DPSS tapers above.
\begin{figure*}[htbp]
\centering
\subfloat{%
  \hspace{4.5cm}\includegraphics[width=0.75\textwidth]{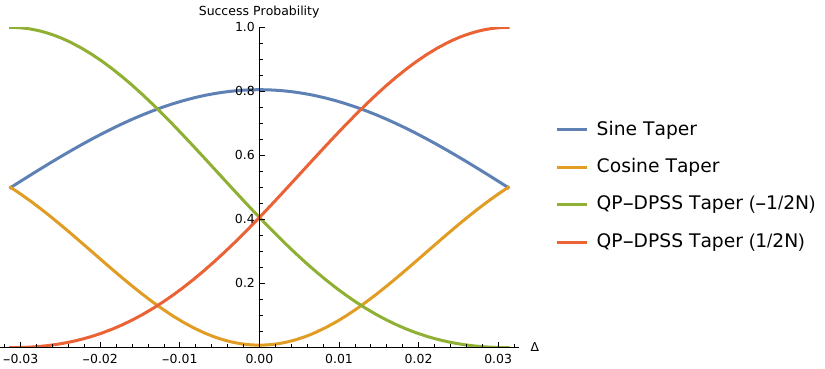}
  }
    \caption{We plot the probability of the $\phi^{-1/2N}$ (green), $\phi^{1/2N}$ (red), $\phi^\text{sin}$ (blue), and $\phi^\text{cos}$ (orange) tapers to output the closest phase estimate as a function of $\Delta$, for $N=2^5$. Both cosine and sine tapers achieve 0.5 at $\Delta=\pm1/2N$. This means with unit probability one of the two closest estimates will be returned. }
    \label{fig:sinevscosine}
\end{figure*}

\subsection{Average-case Optimal Tapers}\label{sec:average-case}

From the development above, it is important to note that in general the value of $\Delta$ is not known \textit{a priori}. Moreover, since $\Delta$ depends on the phase $\theta$, it can differ for distinct phases, particularly when the input state is a superposition of eigenvectors of $U$. Since we are interested in the coherent QPE problem, the same taper must work well for all values of $\Delta$ in order to be useful. Therefore, in this subsection, we focus on finding the optimal taper that works best on average. 

The optimization problem corresponding to this is
\begin{align}
    1 - \overline{\epsilon}^\text{opt} \coloneqq \max_{|\phi\rangle \in \mathcal{H}_{\text{t}}} \quad&  \mathbb{E}_{\Delta \sim \mathcal{D}}\left[ \sum_{j=-K}^K\left|\hat{\phi}\left(\Delta + \frac{j}{N}\right)\right|^2\right] \notag\\
    \text{subject to} \quad &  \sum_{n=0}^{N-1}\left|\phi[n]\right|^2 = 1,\label{eq:org-opt-prob-avg}
\end{align}
where $\overline{\epsilon}^\text{opt}$ denotes the optimal average error probability. In addition, let us denote the average success probability for any given taper $|\phi\rangle$, that is, the objective function above, as follows:
\begin{equation}
    1 - \overline{\epsilon}^{\phi} \coloneqq \mathbb{E}_{\Delta \sim \mathcal{D}}\left[ \sum_{j=-K}^K\left|\hat{\phi}\left(\Delta + \frac{j}{N}\right)\right|^2\right]. 
\end{equation}
The Lagrangian, $\overline{\mathcal{L}}$, for the modified optimization problem in~\eqref{eq:org-opt-prob-avg} is then
\begin{multline}
\overline{\mathcal{L}}\left(|\phi\rangle, \lambda\right)
    =  \mathbb{E}_{\Delta \sim \mathcal{D}}\left[\sum_{j = -K}^{K}\left|\hat{\phi}\left(\Delta + \frac{j}{N}\right)\right|^2\right]\\
    - \lambda\left(\sum_{n=0}^{N-1}\left|\phi[n]\right|^2 - 1\right),\label{eq:avg-case-mainLagrange}
\end{multline}
where $\mathcal{D}$ is the uniform distribution defined over the interval $\left[-\sfrac{1}{2N}, \sfrac{1}{2N}\right]$. 
Then solving for the above Lagrangian, we get the following eigenvalue equation (see App.~\ref{app:averagecase} for a detailed derivation):
\begin{multline}
    \frac{1}{N}\sum_{n=0}^{N-1} \frac{\sin\left(\pi(m - n)(2K+1)/N\right)}{\pi(m-n)}\phi[n]\\
    =\lambda\phi[m];\quad  \forall m \in\{0, \ldots, N-1\}. \label{eq:avg-case-eigenvalue_full} \;
\end{multline}

We observe that this eigenvalue equation appears commonly in the field of classical signal processing and that DPSSs are eigenvectors of this equation. 
Now, with a similar reasoning as we used previously for the ideal case, the eigenvector $|\phi\rangle$ that maximizes the objective function,  $\overline{\mathcal{L}}$ (i.e., average success probability),  is the one with maximum eigenvalue, denoted by $\lambda_{\max}$. Hence, the optimal eigenvector is the DPSS associated with $\lambda_{\max}$, and the optimal value itself equals this maximum eigenvalue:
\begin{equation}
1 - \overline{\epsilon}^\text{opt} =  1 - \overline{\epsilon}^\text{DPSS} = \lambda_{\max}.    
\end{equation}
Thus, using the DPSS taper achieves the optimal average success probability, equal to $\lambda_{\max}$. Also, note that when we mention DPSS from here on, we refer to the DPSS with $\lambda_{\max}$. 
For completeness, we provide a brief overview of DPSSs from the classical signal analysis point of view in App.~\ref{app:classicaldpss}. 

\begin{rmk}An important thing to note is that the derivation of DPSSs from the classical signal processing point of view (provided in App.~\ref{app:classicaldpss}) relies heavily on the fact that the spectrum of frequencies is continuous. In contrast, we have a discrete frequency spectrum in the case of QPE. However, by taking the average of $\Delta$ over the uniform distribution, we effectively transform the problem to a continuous one.  
Thus, the average-case optimization problem has the same form as the classical continuous-frequency case, and as a result, the DPSS taper turns out to be the optimal taper.
\end{rmk}

We now state an important theorem from the classical signal processing literature~\cite{karnik2021improved} that provides (to date) the most stringent non-asymptotic bounds on the eigenvalues of the DPSS kernel in \eqref{eq:avg-case-eigenvalue_full}. Note that this kernel is also known in the literature as the DPSS kernel. 
The key idea is to use this theorem to derive a non-asymptotic bound on the number of additional ancilla qubits, i.e., $m$, needed when employing the DPSS taper to achieve the average success probability $1 - \overline\epsilon^\text{DPSS}$ of at least $1-\overline{\epsilon}$, for some $\overline{\epsilon} > 0$. Since $1 - \overline{\epsilon}^\text{opt} =  1 - \overline{\epsilon}^\text{DPSS}$, we also directly have
\begin{equation}
    1 - \overline{\epsilon}^\text{opt} \geq 1 - \overline{\epsilon}.
\end{equation}

\begin{thm}[{\cite[Corollary 1]{karnik2021improved}}]\label{th:Theorem1}
    For any $N,K \in \mathbb{N}$ such that $K \in \left\{0, \ldots, N/2 - 1\right\}$, the maximum eigenvalue $\lambda_{\max}$ of the DPSS kernel in \eqref{eq:avg-case-eigenvalue_full} satisfies
    \begin{align}\label{eq:max-DPSS-eigenvalue-lower-bound}
        \lambda_{\max} \geq 1 - \min&\left\{8\exp\left[-\frac{2K-1}{\frac{2}{\pi^2}\log(4N)}\right]\right.,\nonumber\\
        &\left.10\exp\left[-\frac{2K-6}{\frac{2}{\pi^2}\log\left(100K+75\right)}\right]\right\}.
    \end{align}
\end{thm}
As the maximum eigenvalue of the DPSS kernel in \eqref{eq:avg-case-eigenvalue_full} is equal to the DPSS average success probability, the bound on the former from the theorem given above naturally bounds the latter. Therefore, one way to infer the above theorem statement is as follows: when the DPSS taper is used, the tQPE algorithm will output one of the $2K+1$ phase estimates closest to $\theta$ with an average success probability lower-bounded by the quantity on the right-hand side of \eqref{eq:max-DPSS-eigenvalue-lower-bound}.

Applying Thm.~\ref{th:Theorem1} for our case, i.e., for the tQPE algorithm that uses the DPSS taper, we obtain the following result that holds for all $N, K \in \mathbb{N}$ which we term as the non-asymptotic regime. 
\begin{thm}[Non-asymptotic]\label{Corollary1}
   Let $\delta \in (0, 1]$. To guarantee that the tQPE algorithm outputs a phase estimate that lies within a $\delta$ distance of the true phase, with an average success probability of at least $1 - \overline{\epsilon}$, it suffices to use
\begin{equation}\label{eq:non-asymptotic-scaling}
m = \left\lceil \log_2\left(\left\lceil 175\left(\log\left(10/\overline{\epsilon}\right) + 1\right)^2\right\rceil + 1\right)\right\rceil + 1 \; .
\end{equation}
additional ancilla qubits.
\end{thm}
\begin{proof}
    We employ the tQPE algorithm that uses the DPSS taper and outputs one of the $\delta$-close $p$-bit phase estimates. Here $p = \ell + m$ and $\ell$ is the minimum positive integer such that $\delta \geq 1/2^{\ell+1}$. This implies that there are $2^m$ grid points (phase estimates) that are within $\delta$ distance to the true phase. 
    
    In Thm.~\ref{th:Theorem1}, we set $N = 2^p$ and $K = 2^{m-1} - 1$ so that we use the bound from Thm.~\ref{th:Theorem1} to lower-bound the DPSS average success probability for outputting one of the $2K+1$ $p$-bit phase estimates that are within $\delta$ distance to the true phase. Since we want
    \begin{equation}
        1 - \overline\epsilon^\text{DPSS} = \lambda_{\max} \geq 1 - \overline{\epsilon},
    \end{equation} 
    we have the following directly from~\eqref{eq:max-DPSS-eigenvalue-lower-bound}:
    \begin{align}\label{eq:error}
        \overline{\epsilon} \geq \min&\left\{8\exp\left[-\frac{2K-1}{\frac{2}{\pi^2}\log(4N)}\right],\nonumber\right. \\
        &\left.10\exp\left[-\frac{2K-6}{\frac{2}{\pi^2}\log(100K+75)}\right]\right\}.
    \end{align}
    Without loss of generality, we choose to work with the second argument of the $\min$ function in the inequality above, i.e., 
    \begin{equation}
  \overline{\epsilon} \geq 10\exp\left[-\frac{ 2K - 6}{\frac{2}{\pi^2}\log(100K + 75)}\right] \; .\label{eq:before-cor-log-flip}
\end{equation}

Taking logarithms and flipping the sign of both sides, we find
\begin{equation}
  \log(10/\overline{\epsilon}) \leq \frac{ 2K - 6}{\frac{2}{\pi^2}\log(100K + 75)}.\label{eq:eps}
\end{equation}

By assuming that $K \geq 1$, we first provide a lower bound for the right-hand side of the above inequality in the following way:
\begin{align}
    \frac{ 2K - 6}{\frac{2}{\pi^2}\log(100K + 75)} &\geq \frac{2K - 6}{\frac{2}{\pi^2}\log(175K)} \\
    &\geq \frac{K-3}{\log\left(175K\right)} \\
    &= \frac{1}{175}\left(\frac{175K}{\log\left(175K\right)}\right) - \frac{3}{\log(175K)} \\
    &\geq \frac{1}{175}\left(\frac{175K}{\log\left(175K + 1\right)}\right) - 1 \\
    &\geq \frac{\sqrt{175K}}{175}- 1 \\
    &= \sqrt{\frac{K}{175}} - 1 \; ,
\end{align}
where the second-to-last inequality comes from the following standard logarithmic inequality:
\begin{equation}
    \frac{x}{\log(x+1)} \geq \sqrt{x+1} \text{ for $x \geq -1$.}
\end{equation}

Now, by enforcing the inequality, 
\begin{equation}
    \log\left(10/\overline{\epsilon}\right) \leq \sqrt{\frac{K}{175}} -1,
\end{equation}
we obtain $K \geq 175\left(\log\left(10/\overline{\epsilon}\right) + 1\right)^2$. 
In other words, as long as $K$ is at least $\lceil 175\left(\log\left(10/\overline{\epsilon}\right) + 1\right)^2\rceil$, we are guaranteed that  \eqref{eq:before-cor-log-flip} holds, and the error is bounded by $\overline{\epsilon}$. 
Recall that $2^m$ of the $p$-bit phase estimates are $\delta$-close. 
Setting $K = 2^{m-1}-1$, we only need 
\begin{equation}
    m =  \left\lceil \log_2\left(\left\lceil 175\left(\log\left(10/\overline{\epsilon}\right) + 1\right)^2\right\rceil + 1\right)\right\rceil + 1
\end{equation}
additional qubits to ensure that our tQPE algorithm has average success probability at least $1-\overline{\epsilon}$. This concludes the proof. 
\end{proof}

\begin{rmk}
Although the analytical proof for Thm.~\ref{Corollary1} implies that $K$ needs to be of order $\log^2(1/\overline{\epsilon})$ to achieve an average success probability of at least $1-\overline{\epsilon}$, for all practical purposes $K$ only needs to be of order $\log(1/\overline{\epsilon})$. For example, for all $\overline{\epsilon}\ge10^{-81}$, we have that $K\le 192$, and as a result, we can replace $K$ inside the logarithm in~\eqref{eq:eps} with 192 and find that $K=\left\lceil \log(10/\overline{\epsilon})\right\rceil +3$ is sufficient.
\end{rmk}

Since Thm.~\ref{th:Theorem1} holds for any $N \in \mathbb{N}$, our result described in Thm.~\ref{Corollary1} also holds for any $N$ including the asymptotic regime where $N\to \infty$ and $\overline{\epsilon} \to 0$. Consequently, any future improvement or tightening of the bound in Thm.~\ref{th:Theorem1} will directly yield a correspondingly sharper and better bound in Thm.~\ref{Corollary1}. If instead we focus on the asymptotic regime, stronger bounds for $m$ can be obtained. More precisely, in this setting, we fix a target precision $\delta \in (0, 1/8]$ and study the limit where the DPSS average success probability approaches 1, i.e., $1 - \overline{\epsilon} \to 1$ (or equivalently, $\overline{\epsilon} \to 0$). The following theorem states our result in this asymptotic case. 

\begin{thm}[Asymptotic]\label{thm:asymptotic-K-scaling} 
    Let $\delta \in (0, 1/8]$. To guarantee that the tQPE algorithm outputs a phase estimate that lies within a $\delta$ distance of the true phase, with an average success probability of at least $1 - \overline{\epsilon}$, it suffices to use
    \begin{equation}
        m \sim \left \lceil\log_2\log\frac{1}{\overline{\epsilon}}\right \rceil
    \end{equation}
additional ancilla qubits.
\end{thm}
\begin{proof}
Here as well we use the tQPE algorithm with the DPSS taper, and the notations are similar to those used in the proof of Thm.~\ref{Corollary1}.

For this proof, we use the result of~\cite{slepian1978prolate} that gave asymptotic expressions for the eigenvalues of the DPSS kernel. In~\cite{slepian1978prolate}, the author denotes the half-bandwidth of the DPSS taper in the frequency domain by $W$ and so to keep the notation consistent with this paper that corresponds to $W = \delta - 1/(2N)$. 
From~\cite{slepian1978prolate}, for a fixed $W$, we have
\begin{equation}
    1 - \lambda_{\max} \sim C(W) N^{1/2} e^{-\gamma(W)N},
\end{equation}
where
\begin{equation}
    C(W) = \pi^{1/2} 2^{9/4} \,\alpha(W)^{1/4} (2 - \alpha(W))^{-1/2}.
\end{equation}
Here $\alpha(W)$ and $\gamma(W)$ are given by
\begin{align}
    \alpha &= 1 - \cos 2\pi W , \\
    \gamma &= \log\!\left(1 + \frac{2\sqrt{\alpha}}{\sqrt{2} - \sqrt{\alpha}}\right).
\end{align}

By the definition of $\sim$, for every $\eta\in(0,1)$ there exists a positive number $N_0$ such that for all $N \geq N_0$, we have
\begin{equation}
    1 - \lambda_{\max} \leq  (1+\eta)C(W) N^{1/2} e^{-\gamma(W)N}.\label{eq:finite-W-proof-lambda0}
\end{equation}

Now since $W = (2^{m}-1)/(2N)$ and $N = 2^{\ell+m}$, we have $NW = (2^{m}-1)/2 =: x$. Plugging $N = x/W$ into~\eqref{eq:finite-W-proof-lambda0}, we get
\begin{equation}
    1 - \lambda_{\max} \leq A(W, \eta) x^{1/2} e^{-\tau(W)x},
\end{equation}
where, again, for brevity, we introduce the following two quantities:
\begin{align}
    A(W, \eta) \coloneqq \frac{(1+\eta)C(W)}{\sqrt{W}}, \quad
\tau(W)\coloneqq \frac{\gamma(W)}{W}
\end{align}

We want that $1 - \lambda_{\max} \leq \overline{\epsilon}$. This implies that we want to find an $x$ such that the following holds:
\begin{equation}
    A(W, \eta) x^{1/2} e^{-\tau(W)x} \leq \overline{\epsilon}\label{eq:finite-W-proof-cond-x-e}.
\end{equation}
Let us pick the following $x$:
\begin{equation}
    x = \frac{1}{\tau(W)}\left( \log(\sfrac{A(W, \eta)}{\overline{\epsilon}}) + \frac{1}{2}\log\log(\sfrac{A(W, \eta)}{\overline{\epsilon}}) \right)\label{eq:finite-W-proof-x-value}.
\end{equation}

Now let us check if this value of $x$ satisfies the condition~\eqref{eq:finite-W-proof-cond-x-e}. We begin by plugging this value in~\eqref{eq:finite-W-proof-cond-x-e}:
\begin{equation}
    A(W, \eta) x^{1/2} e^{-\tau(W)x}
    =\overline{\epsilon}  \sqrt{\frac{x}{\log(\sfrac{A(W, \eta)}{\overline{\epsilon}})}}
\end{equation}
Now using the inequality
\begin{equation}
    \sqrt{\frac{x}{\log(\sfrac{A(W, \eta)}{\overline{\epsilon}})}} \leq \sqrt{\frac{2}{\tau(W)}},
\end{equation}
we have
\begin{equation}
    A(W) x^{1/2} e^{-\tau(W)x} \leq \overline{\epsilon} \sqrt{\frac{2}{\tau(W)}}.
\end{equation}
Furthermore, since for all $W \in (0, 1/2]$, we have $\tau(W) \geq 2 \pi$, we can say that
\begin{equation}
    A(W,\eta) x^{1/2} e^{-\tau(W)x} \leq \overline{\epsilon}.
\end{equation}
This completes the check that the value of $x$ that we picked satisfies the condition~\eqref{eq:finite-W-proof-cond-x-e}.

Since $x = (2^{m}-1)/2$, any integer
\begin{multline}
    m \geq \Bigg\lceil 
        \log_2\!\Bigg(
            \frac{1}{\tau(W)} 
            \Big[ 
                \log\!\left(\sfrac{A(W, \eta)}{\overline{\epsilon}}\right) \\
                + \frac{1}{2}\log\!\log\!\left(\sfrac{A(W, \eta)}{\overline{\epsilon}}\right) 
            \Big]
        \Bigg)
    \Bigg\rceil
\end{multline}
guarantees $1 - \lambda_{\max} \le \overline{\epsilon}$. Now for all $W\in (0, 1/4]$, we bound $A(W, \eta)$ from above as
\begin{equation}
    A(W, \eta) \leq 8\sqrt{2}\pi .
\end{equation}
With this, we finally get a fully explicit sufficient bound for $m$ in terms of $\overline{\epsilon}$ and $W$:
\begin{align}
m \!\ge\! \Bigg\lceil 
\log_{2}\!\Bigg(
\frac{1}{\tau(W)}
\Big[
\log\!\Big(\sfrac{8\sqrt{2}\pi}{\overline{\epsilon}}\Big)
\!+\! \frac{1}{2} \log\!\log\!\Big(\sfrac{8\sqrt{2}\pi}{\overline{\epsilon}}\Big)
\Big]
\Bigg)
\Bigg\rceil.
\end{align}
Since $\tau(W) \geq 2\pi $ for all $W \in (0, 1/2]$, asymptotically as $\epsilon \rightarrow 0$, we finally have
\begin{equation}
    m \sim \left \lceil\log_2\log\frac{1}{\overline{\epsilon}}\right \rceil
\end{equation}
in the finite-$W$ or finite-$\delta$ (precision) regime. This concludes the proof.
\end{proof}

Above, we have given both asymptotic and non-asymptotic bounds on the DPSS average error probability $\overline{\epsilon}$. But the true error probabilities can be much better than suggested by these bounds. In Fig.~\ref{fig:DPSSeigenvalues} we show the average error probabilities for some values of $(\ell, m)$. We observe that already with $m=6$, the average error probabilities below $10^{-80}$ are achieved. With $m=4$, the error probability is below $10^{-17}$.

\begin{figure}[htbp]
\includegraphics[width=0.5
\textwidth]{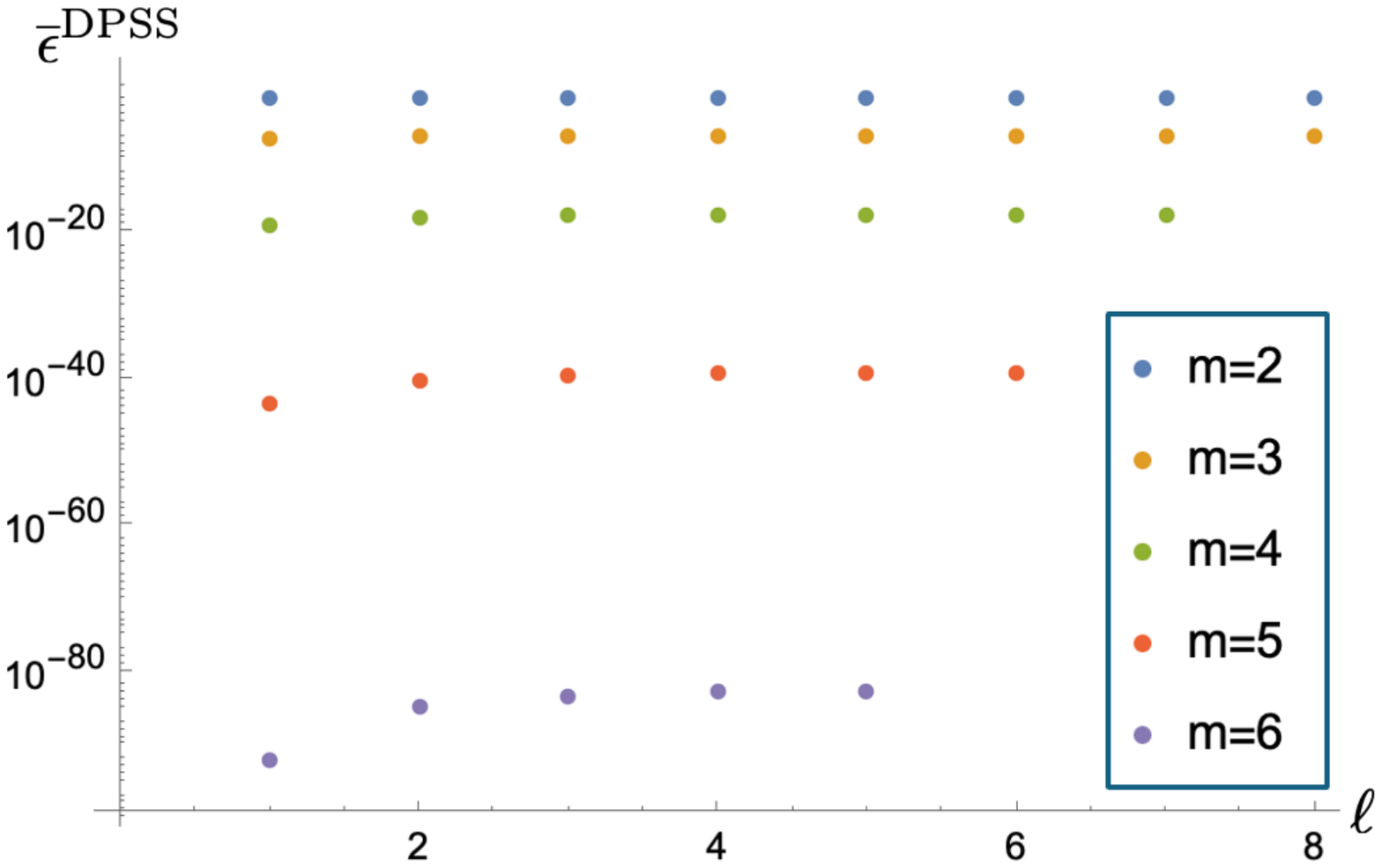}
\caption{
DPSS average error probability as a function of $\ell$ for values of $m\le 6$. We observe that the error probabilities quickly converge with increasing $\ell$. Calculations were performed with Mathematica by setting \texttt{WorkingPrecision}=250.
}
\label{fig:DPSSeigenvalues}
\end{figure}

\subsection{Analysis of worst-case error}

In the discussion above, we optimized with respect to the average over the unknown offset $\Delta$, which led us naturally to the DPSS taper: it is the taper that does best on average when the true phase can sit anywhere between two grid points. But an average guarantee does not tell us what happens at the worst-case (most unfavorable)~$\Delta$, which we denote it as~$\Delta^*$.
Ideally, we would like to find a taper that has the smallest error probability for its worst-case $\Delta^*$. 
For most tapers used in classical signal processing (including the tophat and DPSS tapers), $\Delta^*=\pm 1/2N$ corresponds to their worst-case offset, where the corresponding error probability reaches its maximum. 
However, in general, different tapers have different worst-case $\Delta^*$, which makes the corresponding optimization problem too hard to solve. In this section, we analyze the performance of the DPSS taper in its worst case, which, as mentioned above, is at $\Delta^*=\pm1/2N$. We find that although not optimized for the worst-case scenario, the DPSS taper performs well and is optimal asymptotically. 

To perform such a comparison, we first restate the kernel of the DPSS taper which coincides with the eigenvalue equation shown in \eqref{eq:avg-case-eigenvalue_full}:
\begin{multline}
    \frac{1}{N}\sum_{n=0}^{N-1} \frac{\sin\left(\pi(m - n)(2K+1)/N\right)}{\pi(m-n)}\phi[n]\\
    =\lambda\phi[m];\quad  \forall m \in\{0, \ldots, N-1\}.  \;
\end{multline}
Recall that the DPSS taper is the eigenvector of the DPSS kernel with the largest eigenvalue.
When considering the success probability of the DPSS taper in the worst case, where $\Delta = \pm 1/2N$, we let $|\phi\rangle$ denote the DPSS taper and adapt the calculations in App.~\ref{app:averagecase} to include an additional complex rotation $e^{2\pi i (n' - n)\Delta}$ that translates the phase estimates in the frequency domain by $\Delta = \pm 1/2N$ to obtain the following expression:
\begin{align}
    &\sum_{j = -K}^{K}\left|\hat{\phi}\left(\Delta + \frac{j}{N}\right)\right|^2 \notag\\
    &=\frac{1}{N}\sum_{n,m = 0}^{N-1} e^{2 \pi i (m-n)\Delta} \notag \\
    & \hspace{2em} \times \frac{\sin\left(\pi (m-n)\left(2K+1\right)/N\right)}{\sin\left(\pi (m-n)/N\right)}\phi^*[n]\phi[m] \label{eqn:worst-case-DPSS-kernel-1}\\
    &= \frac{1}{N}\sum_{n,m=0}^{N-1}  \cos\left(\pi (m-n)/N\right)\notag\\
    &\hspace{2em}\times\frac{\sin\left(\pi(m - n)\left(2K+1\right)/N\right)}{\sin\left(\pi (m-n)/N\right)}\phi^*[n]\phi[m],\label{eqn:worst-case-DPSS-kernel}
\end{align}
where the last equality comes from substituting $\Delta = \frac{1}{2N}$ and observing how the complex terms from the complex exponential annihilate each other when we run over all summands indexed by $m$ and $n$.
In other words, the success probability for the DPSS taper in the worst-case would then be the quantity shown in \eqref{eqn:worst-case-DPSS-kernel}.

From Thm.~\ref{th:Theorem1}, we make the observation that the maximum eigenvalue, or the success probability, of the average-case optimal DPSS taper tends to 1 as $N\rightarrow \infty$.
This is well-aligned with our intuition because we would expect our DPSS taper to have greater spectral concentration in the central lobe in the frequency domain.   
At the same time, we note that the worst-case success probability of the DPSS taper also converges to 1 when $N \rightarrow \infty$.
This is a result of the complex exponential term in \eqref{eqn:worst-case-DPSS-kernel-1} approaching 1 as $N \rightarrow \infty$, resulting in the worst-case DPSS success probability expression approaching the average-case success probability in the same limit.
Numerical evidence in Fig.~\ref{fig:DPSSs} suggests that the success probability of DPSS taper in the worst case is still $O(\epsilon)$.

\begin{figure*}[htbp]
\centering
\includegraphics[width=1.00
\textwidth]{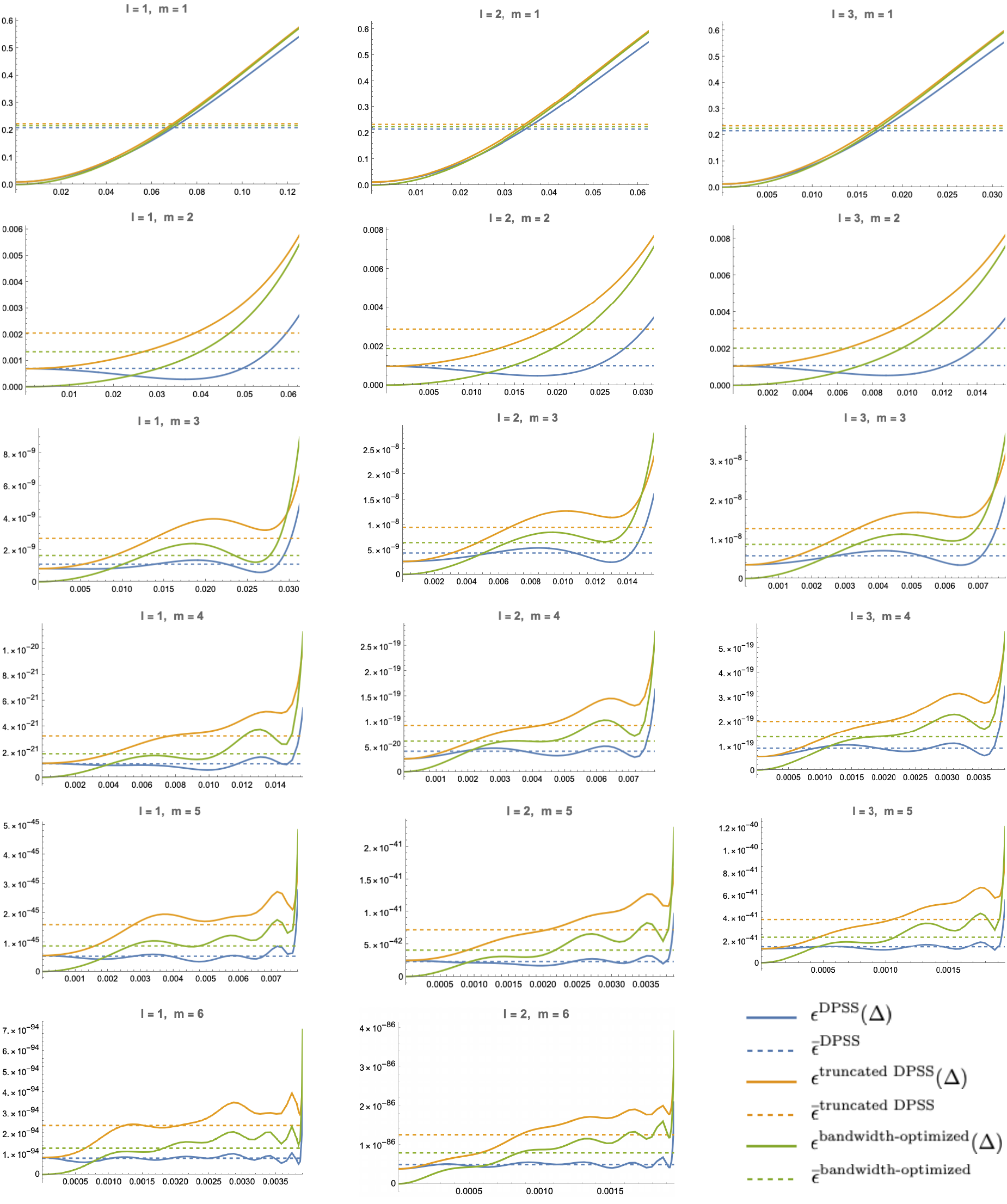}
\caption{
Error probability of the various tapers for $\ell=1,\dots,3$ and $m=1,\dots,6$. The solid lines denote error probability as a function of $\Delta$ from 0 to $1/(2N)$, whereas the dashed lines of the same color are the corresponding averages in this interval. DPSS tapers are described in Sec.~\ref{sec:average-case}, truncated DPSS tapers in Sec.~\ref{sec:truncatedDPSS}, and bandwidth-optimized tapers in Sec.~\ref{sec:optimizedtaper}. We find that the average error of the bandwidth-optimized taper is always less than twice the error of the DPSS taper. Calculations were performed with Mathematica by setting \texttt{WorkingPrecision}=250. 
}
\label{fig:DPSSs}
\end{figure*}

\subsection{Comparison of tapers}\label{sec:numerics}

In this section, we provide numerics for three of the tapers we have discussed so far: the DPSS taper, the sine taper, and the tophat taper.  In Figs.~\ref{fig:l3m4K7_comp},~\ref{fig:l3m3K3_comp},~\ref{fig:l3m2K1_comp},~and~\ref{fig:l3m1K0_comp}, we plot the success probability of each taper as a function of the offset $\Delta$, i.e., the distance between the true phase and the grid point closest to it, for various values of $m$. The DPSS taper is optimal for the average case.
The tophat taper is designed to output the true phase with unit probability when the true phase coincides with a grid point and the sine taper has the property that it outputs one of the two nearest estimates with unit probability when the true phase is exactly between two grid points. Both of these properties are confirmed by the plots. Although the DPSS kernel does not succeed with unit probability for any $\Delta$ it does perform well for the entire range, which is consistent with the fact that it is optimal for the average case. 

\begin{figure*}[h!]
\centering
\subfloat[$m = 4, K = 7$]{%
  \includegraphics[width=0.56\textwidth]{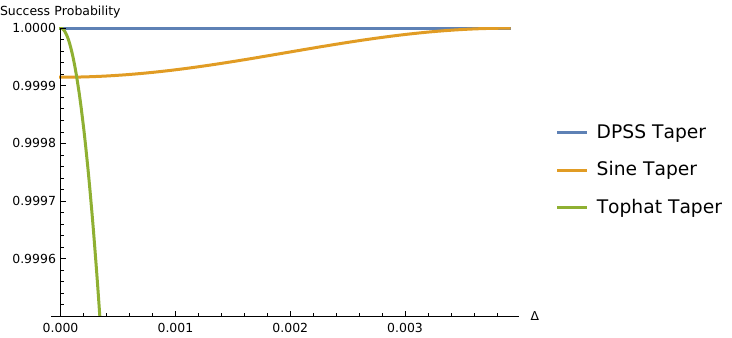}%
  \label{fig:l3m4K7_comp}%
}\quad 
\subfloat[$m =3, K = 3$]{%
  \includegraphics[width=0.42\textwidth]{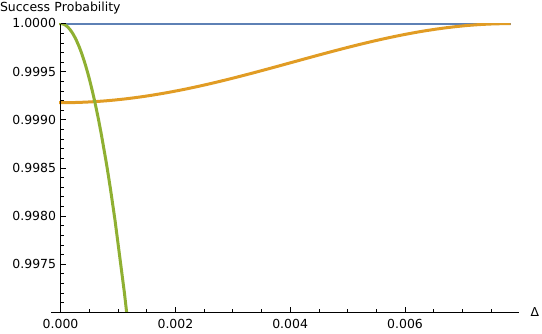}%
  \label{fig:l3m3K3_comp}%
}\\
\subfloat[$m = 2, K = 1$]{%
  \includegraphics[width=0.42\textwidth]{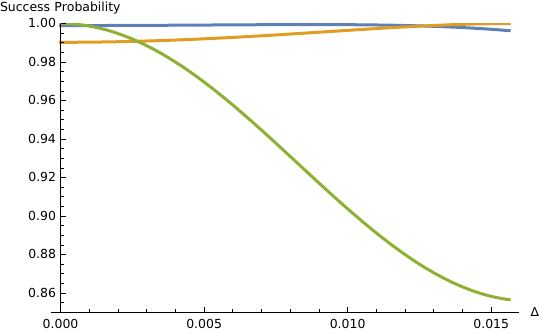}%
  \label{fig:l3m2K1_comp}%
}\quad
\subfloat[$m = 1, K = 0$]{%
  \hspace{2.5cm}\includegraphics[width=0.42\textwidth]{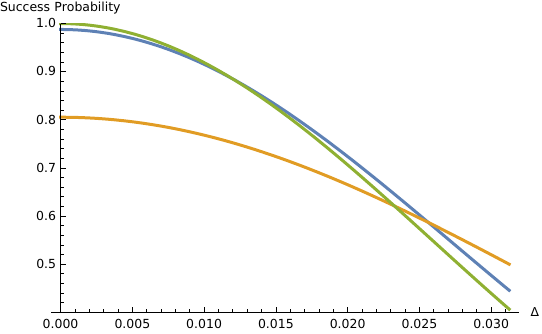}%
  \label{fig:l3m1K0_comp}%
}
\caption{Comparison between the DPSS taper, the sine taper, and the tophat taper.}
\end{figure*}

In Fig.~\ref{fig:DPSSs}, we plot the failure probability of the DPSS taper as a function of $\Delta$ for various values of $l$ and $m$ (solid blue lines). In contrast to Figs.~\ref{fig:l3m4K7_comp},~\ref{fig:l3m3K3_comp},~\ref{fig:l3m2K1_comp},~and~\ref{fig:l3m1K0_comp}, we do not plot the sine and tophat tapers.
We see that the error profile of the DPSS tapers has more structure compared to the tophat and sine tapers. One trend is that for increasing $m$ the error probability concentrates around its average value for most $\Delta$ except near $\Delta=1/(2N)$, where it shoots up. But even in that region, the error appears to be $O(\epsilon)$. In a sense, for large $m$, i.e., high success probability, the DPSS taper does very well almost uniformly (besides being optimal on average).

\section{Error Analysis for QPE with Uncomputation}\label{sec:uncomputing}

In this section, we bound the error in applications of QPE that require uncomputation of the estimated phase. The analysis below applies to all coherent QPE algorithms. In applications that require coherent QPE, like HHL and Quantum Metropolis Sampling, QPE is used in order to approximately implement a unitary of the form:
\begin{equation}
    cT_{CA} \coloneqq \sum_j T(\theta_j)_{C} \otimes |\psi_j\rangle\langle\psi_j|_A,\label{eq:T-unitary}
\end{equation}
where $A$ and $C$ are two registers (see the left-hand side of Figure~\ref{fig:uncomputing}) and $T(\theta)$ is a unitary operator parametrized by $\theta\in[0,2\pi)$. Furthermore, $\{|\psi_j\rangle\}_j$ are the eigenvectors of some unitary $U$ with the corresponding eigenvalues $\{e^{2\pi i \theta_j}\}_j$. We will assume that $T$ is $L_T$-Lipschitz continuous for some $L_T > 0$:
\begin{align}
    \label{eq:lipschitz}
    \| T(\theta)-T(\theta')\| \le L_T \min(|\theta-\theta'|,2\pi - |\theta-\theta'|).
\end{align}
\begin{figure}
    \centering
    \includegraphics[width=0.8\linewidth]{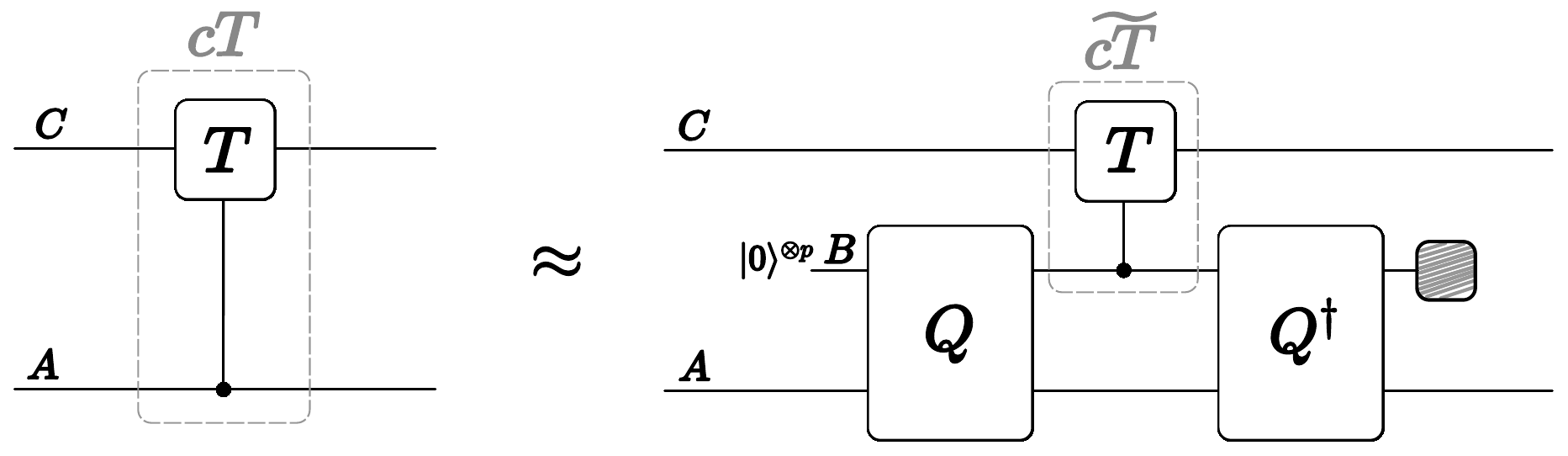}
    \caption{The unitary $cT$ (left) is approximated by the quantum channel (right). Note that $Q$ can be any QPE algorithm. For tQPE, it would include the state preparation unitary for the taper on the ancilla register. }
    \label{fig:uncomputing}
\end{figure} 
We approximately implement the unitary $cT$ using the following algorithm (see the right-hand side of Figure~\ref{fig:uncomputing}):
\begin{enumerate}
    \item Add an ancillary register $B$ and initialize it in the all-zeros state. That is, $|0\rangle^{\otimes p}$ for some $p \in \mathbb{N}$.
    \item Apply a unitary $Q$ (this corresponds to any QPE algorithm, including the tQPE algorithm and the median-based QPE algorithms) on the registers $A$ and $B$ where $A$ holds the system state and $B$ encodes information about the phase. The exact form of this encoding, and how it is later decoded, depends on the specific QPE procedure employed. The general action of $Q$ can be written as
    \begin{equation}
        Q_{BA} \coloneqq \sum_{j, x, y} f_{jxy}  |x\rangle\langle y|_B \otimes |\psi_j\rangle\langle\psi_j|_A\label{eq:qpe-unitary},
    \end{equation}
    where $f_{jxy}$ are the complex numbers determined by the particular QPE algorithm. These amplitudes satisfy, for some $\delta, \epsilon > 0$,
    \begin{equation}
    \sum_{x:|\theta_j-g(x)|> \delta} |f_{jx0}|^2 \leq \epsilon,
\end{equation}
where $g(x): \{0, 1\}^p  \rightarrow [0, 1)$ can be seen as a phase-decoding function associated with the chosen QPE scheme. This function maps computational-basis outcomes $x$ to a numerical phase estimate. The formulation above is general enough to capture all coherent QPE algorithms. For example:
\begin{itemize}
    \item Standard QPE or tQPE: The register $B$ is initialized in $|0\rangle^{\otimes p}$, and the estimate is simply the binary fraction encoded in $x$, i.e.,
        \begin{equation}
            g(x) = \frac{x}{2^p}.
        \end{equation}
        Here and in what follows, $x$ refers both to the binary string in the computational basis as well as the associated integer it represents in binary. The meaning should be clear from context.
    \item Median-based QPE: If $x$ represents a tuple of outcomes $(x^{(1)},\ldots,x^{(m)})$ from $m$ repeated QPE runs each with $l$ bits, then the estimated phase is given by
        \begin{equation}
            g(x) = 
            \operatorname{median}\!\left(
            {x^{(1)}\over 2^l}\,
            \ldots,\,
            {x^{(m)}\over 2^l}
            \right).
        \end{equation}
        Here, the median is taken coherently to suppress outliers among repeated noisy estimates, providing a more robust phase estimate.
\end{itemize}
    \item Apply the following controlled unitary on registers $B$ and $C$ with the control on register $B$:
    \begin{equation}
        \widetilde{cT}_{CB} \coloneqq \sum_x T(g(x))_C \otimes |x\rangle\langle x|_B.
    \end{equation}
    \item Apply $Q^{\dagger}$ on the registers $A$ and $B$.
    \item Trace out register $B$.
\end{enumerate}
Note that although $cT$ and $\widetilde{cT}$ act in a similar way on the $C$ register, the former controls that action on the $A$ register and the latter on the $B$ register. Since the control on the $B$ register is with respect to the computational basis, it is straightforward to implement.

Now, before proceeding with the error analysis of the above algorithm (Theorem~\ref{thm:uncomputing}), we define the quantum channel versions of the unitaries $cT, Q $, and $\widetilde{cT}$:
\begin{align}
    \mathcal{T}(\cdot) & \coloneqq cT (\cdot) cT^{\dagger}\\
    \mathcal{Q}(\cdot) & \coloneqq Q (\cdot) Q^{\dagger}\\
    \widetilde{\mathcal{T}}(\cdot) & \coloneqq \widetilde{cT} (\cdot) \widetilde{cT}^{\dagger}.
\end{align}
We also define an appending channel $\mathcal{A}$, which essentially appends the all-zeros state to the input state, to represent Step~1 of the above algorithm:
\begin{equation}\label{eq:append-channel}
    \mathcal{A}_D(\rho) \coloneqq \rho_D \otimes |0\rangle\langle0|_B.
\end{equation}
Note that in the above equation and in what follows, we use a shorthand $|0\rangle\langle0|_B \equiv |0\rangle\langle0|^{\otimes p}_B$ for simplicity.

With the above notions in place, we now present our QPE uncomputation error analysis in the following theorem.

\begin{thm}\label{thm:uncomputing}
    The following holds:
    \begin{equation}
        \frac{1}{2}\left \Vert \mathcal{T}_{CA} - \operatorname{Tr}_B \circ \mathcal{Q}^{\dagger}_{BA} \circ \widetilde{\mathcal{T}}_{CB} \circ \mathcal{Q}_{BA} \circ \mathcal{A}_{CA} \right\Vert_{\diamond} \leq 2 L_T \delta + 4\epsilon. \label{eq:main-term}
    \end{equation}
    Here, $\operatorname{Tr}_B$ denotes the partial-trace channel, where the register $B$ is traced out.
\end{thm}
The proof of the theorem can be found in App.~\ref{app:uncomputing}.

Eq.~\eqref{eq:main-term} quantifies the error resulting from the uncomputation of QPE combined with the error of approximating the phase in the first place. Let us compare this error to the error introduced solely due the approximation of the phase. For this purpose, we consider the unitary
\begin{align}
    \overline{cT} \coloneqq \sum_j T(\bar{\theta}_j)_{C} \otimes |\psi_j\rangle\langle\psi_j|_A,
\end{align}
where $\bar{\theta}_j$ is a single estimate with $|\bar{\theta}_j-\theta_j|\le \delta$. Let us denote the associated channel by $\overline{\mathcal{T}}$. To implement this unitary one would need access to a QPE algorithm that always outputs a single estimate, which is known not to exist (Here we are ignoring the failure probability, $\epsilon$, which we assume is negligible compared to $L_T \delta$). Nevertheless, in such a case uncomputation would be exact and the only error would be due to the approximation of the eigenphases. In this case
\begin{align}
\label{eq:without-uncomp}
    \frac{1}{2} \left\Vert \mathcal{T} - \overline{\mathcal{T}} \right\Vert_\diamond & \leq \left\Vert cT - \overline{cT} \right\Vert \nonumber\\
    & \hspace{-2cm} = \left\Vert \sum_j T(\theta_j)_{C} \otimes |\psi_j\rangle\langle\psi_j|_A - \sum_j T(\bar{\theta}_j)_{C} \otimes |\psi_j\rangle\langle\psi_j|_A \right\Vert \nonumber\\
    & \hspace{-2cm}= \left\Vert \sum_j \left(T(\theta_j)_{C} -  T(\bar{\theta}_j)_{C}\right) \otimes |\psi_j\rangle\langle\psi_j|_A \right\Vert\\
    & \hspace{-2cm}= \max_j\left\Vert  T(\theta_j)_{C} -  T(\bar{\theta}_j)_{C} \right\Vert\nonumber\\
    & \hspace{-2cm}\leq L_T \delta. \nonumber
\end{align}
The first inequality follows from Lemma~\ref{lem:dia-to-op}, the third equality follows from Lemma~\ref{lem:direct-sum-norm}, and the last inequality follows from the $L_T$-Lipschitz property of the operator $T$.
All in all, from~\eqref{eq:main-term} and~\eqref{eq:without-uncomp}, we can say that the additional error of QPE due to the uncomputing is the same order as the error due to the approximation of the phase.

\section{Randomized Algorithm}
\label{sec:randomization}

In the Sec.~\ref{sec:numerics}, we have seen that the success probability of the tQPE algorithm depends on the true unknown phase.
To eliminate this dependence and achieve minimum error, we utilize a technique presented in~\cite{Lu2023unbiased}, the details of which are reproduced below in the next paragraph for completeness. Although the authors of~\cite{Lu2023unbiased} used randomness to achieve unbiased phase estimation, we utilize their technique to reduce all instances $\Delta$, including the worst-case instance, to the average-case instance, guaranteeing the optimality of our approach. We also note that randomization has been used in other works related to QFT and QPE~\cite{linden2022average,wan2022randomized}.

In the randomized approach, we introduce a random global phase shift $u$ before running QPE by replacing the target unitary $U$ with $e^{2\pi i u}U$. QPE is then performed on this modified unitary to obtain an estimate $\tilde{\theta}$, after which the algorithm outputs $\tilde{\theta} - u$.  
This simple adjustment effectively randomizes the relative alignment between the true phase and the QPE sampling grid, turning any fixed instance into an average-case one while keeping the underlying phase unchanged.

Let $Q$ be a QPE algorithm such that when acting on a phase $\theta$, its error probability is $\epsilon(\theta, Q)$. 
In the randomized algorithm, we sample $u$ uniformly at random from the range $[-1/(2N),1/(2N))$ and pass it to the algorithm $Q$. This algorithm then returns as estimate $\tilde{\theta}$ satisfying $|\tilde{\theta} - (\theta+u)|\le \delta$ with probability at least $1-\epsilon(\theta + u, Q)$, by assumption. The randomized algorithm outputs $\tilde{\theta}_\text{rand} \coloneqq \tilde{\theta} - u$, which satisfies:
\begin{align}
    |\tilde{\theta}_\text{rand}- \theta | &= |\tilde{\theta}-(\theta+u) | \\
    &\le \delta
\end{align}
with success probability $1-\epsilon(\theta+u, Q)$. Therefore, the probability that the estimate $\tilde{\theta}_\text{rand}$ is $\delta$-close to $\theta$ is then given by
\begin{align}
    & \text{Prob}(|\tilde{\theta}_\text{rand}-\theta|\le\delta) \nonumber\\
    &= \mathbb{E}_{u \sim \mathcal{D}}\left[ (1-\epsilon(\theta+u, Q))\right] \\
    &= 1-\mathbb{E}_{u \sim \mathcal{D}}\left[ \epsilon(\theta+u, Q)\right] \\
    &= 1 - \overline{\epsilon}(Q),
\end{align}
where $\overline{\epsilon}(Q)$ is the average error probability of the QPE algorithm $Q$.

This procedure acts as a random self-reduction, effectively reducing any specific instance to an average instance.
As a result, the performance of any randomized QPE algorithm is independent of the true phase $\theta$, and the error probability can be upper-bounded by the average error probability of the corresponding pre-randomized QPE algorithm. Since we have already argued that the DPSS taper has the optimal average error probability, the randomized tQPE algorithm with the DPSS taper is provably optimal in the worst-case as well.

The random phase shift can be implemented in tQPE with minimal additional gate cost and no increase in the number of qubits. We only need to modify the middle part of the tQPE algorithm, given in Eq.~\eqref{eq:selectU}, as:  
\begin{equation}
    \label{eq:selectU_randomized}
    \sum_{n=0}^{N-1} e^{2\pi i n u }\ketbra{n}{n}\otimes U^n.
\end{equation}
In practice, this effect can be achieved by adding a sequence of single-qubit rotations acting on the ancilla qubits since
\begin{equation}
    \label{eq:selectU_randomized}
    \bigotimes_{s=0}^{p-1} \left(e^{2\pi i 2^{s} u Z} \right)_{a_s} = e^{- 2\pi i (N/2)u} \sum_{n=0}^{N-1} e^{ 2\pi i n u} \ketbra{n}{n}. 
\end{equation}
Here we adopt the convention that the ancilla qubits are labeled as $0,1,\dots,p-1$, starting from the bottom ancilla qubit in Fig.~\ref {fig:circuitdiagram}. Recall that the $s$-th ancilla qubit controls the application of $U^{2^s}$ in the deterministic algorithm. The subscript $a_s$ indicates the ancilla register the operator is acting on. The global phase outside the sum on the right-hand side is physically irrelevant.  

\section{Taper preparation and total error of the algorithm}
\label{sec:taperprep}

The important issue of how to efficiently prepare a good taper state is addressed in detail in App.~\ref{app:TaperPreparation} and an example worked out in App.~\ref{app:TaperPreparationExample}. Here we explain the ideas behind that method and use the results to establish the total error of the tQPE algorithm. 

\subsection{Truncated DPSS taper
\label{sec:truncatedDPSS}}

In Sec.~\ref{sec:tQPE}, we have seen that the tQPE algorithm implements the following transformation on the taper states
\begin{equation}
    \ket{\phi} \ket{\psi_\theta} \xrightarrow{\text{tQPE}} \ket{\hat{\phi}_\theta} \ket{\psi_\theta},
\end{equation}
where
\begin{equation}
    \ket{\hat{\phi}_\theta} \coloneqq \sum_{k=0}^{N-1} \hat{\phi}(\theta - k/N) \ket{k} \; .
\end{equation}
In the special case when $\theta=0$, the outcome state
\begin{align}
    \ket{\hat{\phi}} := \sum_{k=0}^{N-1} \hat{\phi}(-k/N) \ket{k}\, 
\end{align}
is the discrete-time Fourier transform of the taper. 
Since the DPSS taper has only $2K + 1$ significant frequency components and the rest are extremely suppressed, it can be safely truncated to these few dominant basis states while introducing a negligible error.  

To see this, let $P_\theta$ be the projector into the subspace spanned by the $2K+1$ computational basis states closest to $\theta$, which for $\theta=0$ corresponds to the central lobe of the DPSS taper. 
With this, we define a truncation of the DPSS taper as
\begin{align}
\label{eq:approxDPSS}
    \ket{\hat{\phi}^\ast} = \frac{P_0 \ket{\hat{\phi}}}{\left \|P_0 \ket{\hat{\phi}} \right\|_2} \; .
\end{align}
By the properties of the DPSS taper we stated above, it holds that $\vert \braket{\hat{\phi}|\hat{\phi}^*}\vert^2=1- \epsilon(0)$, which implies 
\begin{align}
    \ket{\hat \phi^\ast} = \sqrt{1-\epsilon(0)}\ket{\hat \phi} + \sqrt{\epsilon(0)} \ket{\hat \phi^\perp}\,,
\end{align}
where $\ket{\hat \phi^\perp}$ is orthogonal to $\ket{\hat \phi}$. Note that $\epsilon(\theta)$ had been defined in Eq.~\eqref{eq:epstheta}. 
Since QFT is unitary, the same holds for the time-domain tapers, i.e.,
\begin{equation}
    \ket{ \phi^\ast} = \sqrt{1-\epsilon(0)}\ket{ \phi} + \sqrt{\epsilon(0)} \ket{ \phi^\perp}.
\end{equation} 
The idea is then to first prepare $\ket{\hat{\phi}^\ast}$ using brute-force quantum-state-synthesis approaches such as the one presented in~\cite{shende2005synthesis} and then apply QFT to recover the approximate taper in the time domain. In App.~\ref{app:TaperPreparation}, we show that the gate complexity for this state-preparation algorithm is
\begin{equation}
    \mathcal{C}_\phi = O\!\left(\log^2(1/\bar{\epsilon}) +  \log^2(1/\delta)\right).
\end{equation}

When using the approximate taper $\ket{\phi^\ast}$, the total error probability of the tQPE algorithm gets contributions from two sources: (i) the error probability when the exact DPSS taper is used and (ii) the additional error incurred due to the deviation of $\ket{\phi^\ast}$ from the exact DPSS taper. The corresponding ancilla state at the end of the tQPE algorithm
\begin{equation}
    \ket{\phi^*}\ket{\psi_\theta} \xrightarrow{\text{tQPE}} \ket{\hat{\phi}^\ast_\theta} \ket{\psi_\theta}
\end{equation}
is 
\begin{equation}
      \ket{\hat{\phi}^\ast_\theta} =  \sqrt{1-\epsilon(0)}\ket{\hat \phi_\theta} + \sqrt{\epsilon(0)} \ket{\hat \phi^\perp_\theta}.
\end{equation}
The success probability when using the approximate DPSS taper~\eqref{eq:approxDPSS} is then lower-bounded as
\begin{align}
\nonumber
    \bra{\hat{\phi}^\ast_\theta} P_\theta\ket{\hat{\phi}^\ast_\theta} &= (1-\epsilon(0)) \bra{\hat{\phi}_\theta} P_\theta\ket{\hat{\phi}_\theta} + \epsilon(0) \bra{\hat{\phi}^\perp_\theta} P_\theta\ket{\hat{\phi}^\perp_\theta} \\
    &\quad + 2 \sqrt{\epsilon(0) (1-\epsilon(0))}\mathfrak{Re}\bra{\hat{\phi}^\perp_\theta} P_\theta\ket{\hat{\phi}_\theta} \\
    \label{eq:lb1}
    &\ge (1-\epsilon(0)) (1-\epsilon(\theta)) \nonumber \\
    &\quad- 2 \sqrt{\epsilon(0) (1-\epsilon(0))}|\bra{\hat{\phi}^\perp_\theta} P_\theta\ket{\hat{\phi}_\theta}| \; ,
\end{align}
where to get the inequality, we used the properties of the DPSS taper for the first term, set the second term to 0, and replaced the last term with the negative of its absolute value.

Next, we will show that $|\bra{\hat{\phi}^\perp_\theta} P_\theta\ket{\hat{\phi}_\theta}|$ is small. 
Let $P_\theta^\perp = I - P_\theta$.
Since $0=\braket{\hat{\phi}^\perp_\theta | \hat{\phi}_\theta}=\braket{\hat{\phi}^\perp_\theta |(P_\theta+P_\theta^\perp)| \hat{\phi}_\theta}$ it follows that 
\begin{align}
    |\bra{\hat{\phi}^\perp_\theta} P_\theta\ket{\hat{\phi}_\theta}| = |\bra{\hat{\phi}^\perp_\theta} P_\theta^\perp\ket{\hat{\phi}_\theta}| \; .
\end{align}
Using Cauchy–Schwarz inequality, the left-hand side is bounded from above by
\begin{equation}
    \|P_\theta \ket{\hat{\phi}_\theta}\| \|P_\theta \ket{\hat{\phi}_\theta^\perp}\|=\sqrt{1-\epsilon(\theta)} \sqrt{a}
\end{equation}
whereas the right-hand side is bounded from above by
\begin{equation}
    \|P_\theta^\perp \ket{\hat{\phi}_\theta}\| \|P_\theta^\perp \ket{\hat{\phi}_\theta^\perp}\|=\sqrt{\epsilon(\theta)} \sqrt{1-a},
\end{equation}
where 
\begin{equation}
a \coloneqq \|P_\theta \ket{\hat{\phi}_\theta^\perp} \|^2\in[0,1].
\end{equation}
Thus, 
\begin{align}
    |\bra{\hat{\phi}^\perp_\theta} P_\theta\ket{\hat{\phi}_\theta}| &\le \min\!\left(\sqrt{a(1-\epsilon(\theta))},\sqrt{\epsilon(\theta)(1-a)}\right) \\
    &\le \sqrt{\epsilon(\theta) (1-\epsilon(\theta))} \; ,
\end{align}
where for the last inequality, we set $a =  \epsilon(\theta)$, which maximizes the minimum for  $a\in[0,1]$. Substituting the above bound back in Eq.~\eqref{eq:lb1}, we obtain
\begin{align}
\nonumber
    & \bra{\hat{\phi}^\ast_\theta} P_\theta\ket{\hat{\phi}^\ast_\theta}\\
    & \ge (1-\epsilon(0)) (1-\epsilon(\theta)) \nonumber
    \\
    &\quad - 2 \sqrt{\epsilon(0) (1-\epsilon(0))}\sqrt{\epsilon(\theta) (1-\epsilon(\theta))} \\
    &\ge 1-\left(\epsilon(0)+\epsilon(\theta)+2\sqrt{\epsilon(0)\epsilon(\theta)}\right)\label{eq:bound-on-succ-prob-approx-taper}.
\end{align}

Finally, we consider the average success probability, which is
\begin{equation}
    \mathbb{E}_{\theta \sim \mathcal{D}} \left[ \bra{\hat{\phi}^\ast_\theta} P_\theta\ket{\hat{\phi}^\ast_\theta}\right].
\end{equation}
Taking expected value on both sides of~\eqref{eq:bound-on-succ-prob-approx-taper} and using Jensen's inequality $\mathbb{E}[\sqrt{x}]\le \sqrt{\mathbb{E}[x]}$, we get
\begin{align}
       & \mathbb{E}_{\theta \sim \mathcal{D}} \left[ \bra{\hat{\phi}^\ast_\theta} P_\theta\ket{\hat{\phi}^\ast_\theta} \right]  \nonumber \\
       & \ge 1 - \left(\epsilon(0)+\overline{\epsilon}^\text{DPSS}+2\sqrt{\epsilon(0)\overline{\epsilon}^\text{DPSS}}\right) \\
    &\ge 1-4\max(\epsilon(0),\overline{\epsilon}^\text{DPSS}), 
\end{align}
where, as defined before $\overline{\epsilon}^\text{DPSS}$ is the DPSS average error probability.  
In most cases $\epsilon(0)\approx \overline{\epsilon}^\text{DPSS}$. Consequently, if we want to lower-bound the success probability by $1-\epsilon$, choosing $\overline{\epsilon} = \overline{\epsilon}^\text{DPSS}=\epsilon/4$ in the theorems and lemmas in Sec.~\ref{sec:average-case} is sufficient.

\subsection{Optimized bandwidth-limited Taper}
\label{sec:optimizedtaper}

The approximate taper above was obtained by first computing the optimal taper and then chopping its tails off in the frequency domain. This was necessary for the state preparation to be efficient. We can achieve a factor of two improvement by directly optimizing over the band-limited tapers that can be prepared efficiently. This is a strict improvement over the truncated DPSS taper, as the latter is also band-limited in the frequency domain. We start with the optimization problem in the average-case error setting. The cost function is
\begin{align}
p_\text{success} &= \sum_{n,n'=0}^{N-1}\phi[n] C(n,n') \phi[n']^* \\
C(n,n') &= \frac{\sin(\pi(n-n')(2K+1)/N)}{\pi (n-n')}
\end{align}
Let us write the success probability in terms of the Fourier coefficients of the taper. 
\begin{equation}
    \hat{\phi}[k] \coloneqq \frac{1}{\sqrt{N}}\sum_{n=0}^{N-1}\phi[n]e^{-2\pi i nf}, 
\end{equation}
Note that here we are not referring to the discrete-time Fourier transform of Eq.~\eqref{eq:FTofTaper}.
We use angular brackets to indicate that the argument is an integer index. Then
\begin{align}
p_\text{success} &= \sum_{k,k'=0}^{N-1}\hat\phi[k] D(k,k') \hat\phi[k]^* \; ,\\
D[k,k'] &= \frac{1}{N} \sum_{n,n'=0}^{N-1} C[n,n'] e^{2\pi i (n k-n' k')/N} \; .
\end{align}
In the previous approach we chopped off the tails of the taper in the frequency space, hence setting $\hat\phi[k]=0$ for $K<|k|\le N$. Thus the success probability is 
\begin{align}
\label{eq:optimizedsuccessp}
p_\text{success}^\text{truncated} &= \sum_{k,k'=-K}^{K}\hat\phi[k] D[k,k'] \hat\phi[k]^*  e^{2\pi i (n k-n' k')/N}
\end{align}
We can directly maximize this expression. There are some benefits to this approach. First, this success probability is larger than that of the truncated DPSS since it is a maximum over a set that does include it. Second, the size of the eigenvalue problem is of dimension $2K+1\approx 2^m$ as opposed to $N=2^{\ell+m}$. The first downside is that the matrix elements of the new kernel $D[n,n']$ require a double sum over $N$ items leading to $N^2$ complexity if done by brute force. Second, the new kernel $D[k,k']$ might not be as well studied as DPSS and it might be less efficient to compute classically. 

In the numerics for $l+m \le 10$ we found that the average success probability is larger than $1-2\bar{\epsilon}$ which can also be seen in Fig.~\ref{fig:DPSSs}. Consequently, if we want to lower-bound the success probability by $1-\epsilon$ choosing $\bar{\epsilon}=\epsilon/2$ in the theorems and lemmas in Sec.~\ref{sec:average-case} is likely sufficient.

\subsection{Tapers extrapolated from small $l^\ast$}
\label{sec:extrapolatedtaper}

In Sec.~\ref{sec:truncatedDPSS} we implicitly assumed knowledge of the DPSS taper. In practice, a DPSS taper might be hard to compute classically for large values of $p$. Similarly, the calculation of optimal bandwidth-limited tapers in Sec.~\ref{sec:optimizedtaper} might be expensive. Note that these are one-time classical calculations for each $(l,m)$-pair. Nonetheless, in this section we present a method to classically efficiently generate approximate tapers that nevertheless perform very well. 

The idea is to follow the prescription of Sec.~\ref{sec:optimizedtaper} (one could also follow Sec.~\ref{sec:truncatedDPSS} with similar results) with some $\ell^\ast<\ell$ for which we can carry out the classical computation. The outcome of this calculation is a vector of dimension $2K+1$. In those sections we turned this vector to an $N^\ast=2^{\ell^\ast+m}$ dimensional one by appending zeros. Applying QFT then resulted in the optimal taper. Here, instead, we turn this vector to a $N=2^{\ell+m}$ dimensional one and then apply QFT to get the new taper. This taper is suboptimal, since the $2K+1$ dimensional vector was optimized for $(\ell^\ast,m)$. Nevertheless, it performs very well as we can see in Fig.~\ref{fig:taperextrapolation}, even if $\ell\gg \ell^\ast$. 

\begin{figure*}[t!]
\includegraphics[width=2.08\columnwidth]{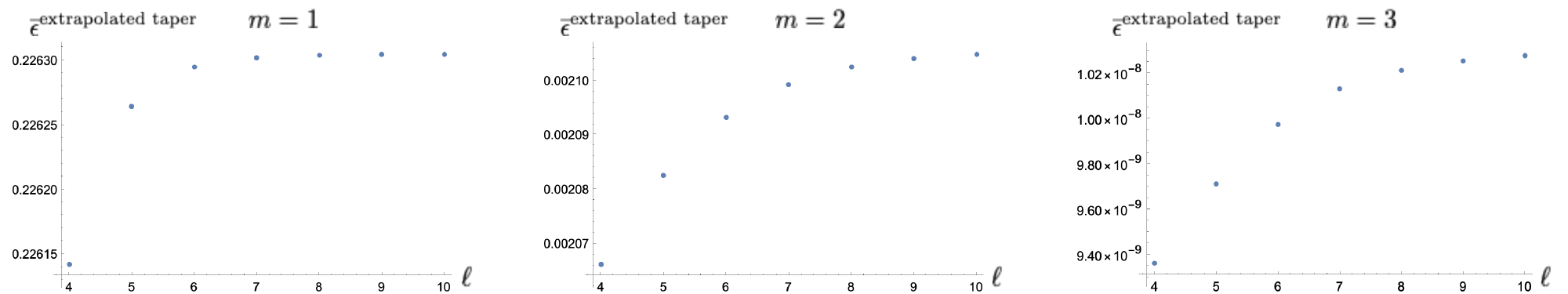}
\caption{Demonstration of the performance of the tapers from Sec.~\ref{sec:extrapolatedtaper}. Plotted are the average errors of tapers for $m=1,2,3$ and $l=4,\dots,10$ that have been obtained with classical computations for the problem size with fixed $\ell^\ast =4$. These tapers are optimal for only $\ell=4$ by construction, but nonetheless perform quite well for larger $\ell$.
}
	\label{fig:taperextrapolation}
\end{figure*}

\section{Discussion}
\label{sec:discussion}

In this manuscript we studied the performance of the QPE algorithm as the input state of the ancilla register (called a taper here) is varied. We found that different objective criteria result in different optimal tapers. By combining QPE with phase randomization, we were able to show that the success probability of outputting an estimate within the desired error bars is maximized by the DPSS taper. 
The DPSS taper is optimal in all parameter regimes, not only asymptotically. 
Given the error probabilities shown in Fig.~\ref{fig:DPSSeigenvalues} it is likely that most applications will require $m\le 4$ additional qubits beyond what is needed for any desired precision.
We also showed how approximations to this taper with similar performance guarantees can be prepared efficiently on a quantum computer. 

Quantum algorithms that require coherent phase estimation, such as the HHL algorithm~\cite{harrow2009quantum} and Quantum Metropolis Sampling~\cite{temme2011quantum}, also require that the estimate of the phase be uncomputed at some point. If this were not the case, incoherent phase estimation would be sufficient, as discarding the ancilla register after coherent QPE renders it equivalent to incoherent QPE. Because tQPE outputs multiple estimates for a particular phase (unlike median boosted QPE that outputs the closest two estimates with high probability), one might worry that the uncomputation step might fail. In Sec.~\ref{sec:uncomputing} we argued that this concern is unfounded: the error of any coherent QPE algorithm due to uncomputation is of the same magnitude as the error due to the approximation of the phase. 
For a precise statement of this result, see Theorem~\ref{thm:uncomputing}.
Note that this theorem is proven in App.~\ref{app:uncomputing} for tQPE without randomization. The proof can be extended to the randomized method described in Sec.~\ref{sec:randomization}; however, the resulting error parameter $\epsilon$ in the upper bound remains a worst-case error rather than the average error $\bar\epsilon$ (a detailed proof is omitted). This distinction is not expected to have practical impact, as the upper bound will typically be dominated by $\delta$ in most applications.

We note that the authors of~\cite{rzkadkowski2017discrete} studied continuous QPE and formulated a different optimization problem than ours. Interestingly, their optimal taper for continuous QPE also resulted in DPSS. So, it would be valuable to explore the similarities and differences between these two optimization problems, which we leave for future work.

We remark that the asymptotic scaling of the number of additional qubits for the DPSS taper compares favorably against the Kaiser taper which is known to also achieve a similar asymptotically scaling of $m$ as our DPSS taper~\cite{berry2022quantifying}. There are other additional factors inside the logarithm for the Kaiser taper which the DPSS taper does not have. 
Furthermore, for the case of the Kaiser taper, its $K$ value scales as $O\left(\log(1/\epsilon)\right) + O\left(\log\log(1/\epsilon)\right)$ with respect to $\epsilon$ in the asymptotic regime as shown in Appendix D of~\cite{berry2022quantifying}.
As far as we know, all qubit-count analysis for the number of qubits required by the Kaiser taper to boost the success probability for QPE has only been done using the asymptotic expression. On the other hand, our work provides a non-asymptotic minimum qubit count for the DPSS taper for boosting the success probability. For the asymptotic case,
as shown in the proof for Thm.~\ref{thm:asymptotic-K-scaling}, the DPSS taper scales as $O\left(\log(1/\epsilon)\right)$ in the asymptotic regime which implies that the DPSS taper has better frequency concentration around $\theta$ than the Kaiser taper in the regime where $N$ is large and $\delta$ is small.
This corresponds well with existing classical signal processing literature where the DPSS taper is known to have maximal frequency concentration in the central lobe out of all taper functions and the Kaiser taper has originally been designed as an approximation to the optimal DPSS taper.

In addition to the asymptotic analysis, we also provide the explicit expressions for $m$ in the non-asymptotic regime. This will be important for the practical implementation of the QPE algorithm in the early fault-tolerant regime. Finally, we would like to mention that we approach the QPE problem from a different angle, i.e., by framing optimization problems. Due to the optimality of the DPSS taper, all tapers must use $\Omega(\log\log(1/\epsilon))$ additional qubits to boost the average success probability to at least $1-\epsilon$. This means that there cannot be any further improvements when using tapers for QPE.
This result complements Thm. 1.3 of~\cite{mande2023tight}, which provides the same lower bound on the $\epsilon$-dependence for any algorithm solving QPE in the worst case. However, while Thm. 1.3 of~\cite{mande2023tight} is general and holds for any QPE algorithm, our lower bound only applies to tQPE with arbitrary tapers.

Due to the flexibility of our formulation, it allows us to find tapers satisfying different optimization criteria. For instance, we find that a sinusoidal taper minimizes the average case error under the additional constraint that for $m=0$ the error probability is zero whenever the true phase happens to be exactly in between two estimates.
This taper has been obtained before in the literature as a result of a different optimization problem~\cite{luis1996optimum, buvzek1999optimal}: it is the taper that minimizes a cost function that penalizes estimates that are further away from the true value as $4 \sin^2((\theta-\theta_\text{est})/2)$. We have arrived at this taper from a completely different perspective.

\textit{Note added:} During the completion of this manuscript, we became aware of a recent numerical study conducted by Greenaway~\emph{et al.}, concluding that QPE performed with the Kaiser taper significantly outperformed QSVT-QPE~\cite{Rall2021} (QPE performed with the QSVT framework)~\cite{greenaway2024case}. In particular, they observed that QPE with Kaiser tapers achieved success probabilities that are orders of magnitude higher than QSVT-QPE with better query complexity in the highly relevant regime where the desired success probability is very close to 1.  Readers can find a more explicit comparison of the query complexity for the QSVT-QPE approach and the tQPE approach with the Kaiser taper in the figures in Ref.~\cite{greenaway2024case}. This finding further emphasizes the importance of utilizing tapers for the QPE problem.

\section{Acknowledgements}

The authors thank Sam Slezak for useful discussions related to the error analysis in Sec.~\ref{sec:uncomputing} and anonymous reviewers for their suggestions that helped improve the presentation of this work. DP and SJST were participants in the 2023 Quantum Computing Summer School at LANL, sponsored by the LANL Information Science \& Technology Institute. DP spent the following fall semester as a GRA at LANL during which he was supported by the Laboratory Directed Research and Development program of Los Alamos National Laboratory (LANL) under project number 20230049DR.
This material is based upon work supported by the U.S. Department of Energy, Office of Science, National Quantum Information Science Research Centers, Quantum Science Center (QSC). YS was funded by the QSC to perform the analytical and numerical analysis and to write the manuscript along with the other authors.
ATS acknowledges initial support from the LANL ASC Beyond Moore's Law project and subsequent support from the Laboratory Directed Research and Development program of Los Alamos National Laboratory (LANL) under project number 20230049DR. 

\bibliographystyle{alpha}
\bibliography{Biblio}

\vfill\eject\newpage

\onecolumngrid

\appendix
\renewcommand\thefigure{\thesection\arabic{figure}}
\renewcommand{\thesection}{\Alph{section}}
\setcounter{figure}{0}
\setcounter{section}{0}

\begin{center}
\large{Supplementary Material for\\``Optimal Coherent Quantum Phase Estimation via Tapering''
}
\end{center}

\section{Approximate Taper Preparation}\label{app:TaperPreparation}

In this section, we focus on preparing the tapers described in Sec.~\ref{sec:taperprep}. We use a modification of the method for preparing broadband Gaussian states presented in~\cite{chowdhury2018improved} to prepare the taper state.
The key idea is that the DPSS taper in the frequency domain, i.e., $\ket{\hat \phi} = U_{\text{QFT}}^\dagger\ket{\phi}$, is narrow-band.
By narrow-band, we mean that it is highly concentrated on just a few grid points, implying that the state, $\ket{\hat \phi}$, has appreciable amplitudes on very few of its basis states. 
Therefore, we can prepare an approximate version of $\ket{\hat{\phi}}$ (say $\ket{\hat{\phi}^\ast}$) using standard methods presented in~\cite{chowdhury2018improved} with very little overhead.
After that, we can perform a quantum Fourier transform on $\ket{\hat{ \phi}^\ast}$ to obtain $\ket{\phi^\ast}$, which closely approximates the DPSS taper $|\phi\rangle$. 
\begin{lem}\label{Lemma1}
The gate complexity of preparing the approximate tapers of Secs.~\ref{sec:truncatedDPSS},\ref{sec:optimizedtaper}, and~\ref{sec:extrapolatedtaper} is
\begin{align}
    \mathcal{C}_\phi &= O\left(\log^2(1/\bar\epsilon) +  \log^2(1/\delta)\right)  
\end{align}
\end{lem}
\begin{proof}
We provide an explicit description of the taper state preparation method. The tapers of Secs.~\ref{sec:truncatedDPSS},\ref{sec:optimizedtaper}, and~\ref{sec:extrapolatedtaper} all have in common the fact that in the frequency domain they have only $N'=2K+1$ non-zero entries. Since the complexity of state preparation only depends on this parameter, the analysis below applies to all these tapers. 
Using the results of Thm.~\ref{Corollary1}, it is sufficient to choose $m = \lceil \log_2(x)\rceil +1$ with $x= \left\lceil 175\left(\log\left(10/\bar\epsilon\right) + 1\right)^2\right\rceil + 1$. We can upper-bound $N' = 2K+1 = 2^m-1$ as 
\begin{align}
    N' & = 2^{\lceil \log_2(x)\rceil +1} - 1 \\
    &\le 4x\; .
\end{align}
Hence $N' = O(\log^2(1/\bar\epsilon))$.
In the frequency domain the taper is centered around 0, and is bandlimited.  
We use the $N'=2^m-1$ amplitudes in the central peak to construct a quantum state on a register of $\log (N'+1)=m$ qubits. This state is centered at the computational basis state 0 and symmetric around it modulo $m$. The amplitude at basis state $2^{m-1}$ is set to zero. 
This state is then prepared using the method described in~\cite{shende2005synthesis}. 
Subsequently, we append $\ell$ qubits initialized to $\ket{0}^{\otimes \ell}$ to the ancilla register to form an $(\ell+m)$-qubit ancilla register, in such a way that the added qubits represent the most significant digits of the computational basis. Now the taper is not centered around 0 anymore; all the nonzero amplitudes are on the computational basis states in the interval $[0,2^m-1]$. To make the taper centered again it suffices to act with the following unitary
\begin{align}
    \prod_{n=m+1}^{p} \text{CNOT}_{m\rightarrow n}
\end{align}
This unitary leaves the computational basis states in the interval $[0,2^{m-1}-1]$ unchanged and adds $2^p-2^m$ to the basis states in the interval $[2^{m-1},2^{m}-1]$. This is exactly the transformation needed to center the taper around 0. Finally, we act with $U_\text{QFT}$ to obtain the taper in the time-domain which is centered around $N/2$ and $N/2+1$. This is the approximate taper.

Now, we analyze the complexity of the taper state preparation protocol.~\cite{shende2005synthesis} introduced a technique for preparing a state on a register of $m$ qubits with a gate complexity of $O(2^m)$. 
Note that their technique requires some classical computation, which we disregard, along with the classical computation of the $2^m-1$ amplitudes of the central lobe of $|\hat{\phi}\rangle$.
The quantum circuit required to center the amplitudes consists of $\ell$ CNOT gates. 
Lastly, we require an application of QFT which has been shown in~\cite{coppersmith1994approximate} to have gate complexity $O\left(\log^2 N\right)=O\left(p^2\right) = O((l+m)^2)$. This gives us an overall gate complexity of $O\left(2^m + \ell^2 \right)$ for the state preparation protocol.We can state the gate complexity of the state preparation protocol $\mathcal{C}_\phi$ in terms of $\delta$ and $\epsilon$ as:
\begin{align}
    \mathcal{C}_\phi &= O\left(\log^2(1/\bar\epsilon) +  \log^2(1/\delta)\right)\; . 
\end{align}
\end{proof}
We note that the state preparation protocol does not require any additional ancilla qubits other than the $p$-qubits for the taper state register.

\subsection{Taper Preparation Circuit Example}
\label{app:TaperPreparationExample}
In this section, we provide a concrete example of how our DPSS taper can be prepared using the truncation method described in App.~\ref{app:TaperPreparation}.
In this example, we assume that $\ell = 2$ and $m = 3$. Recall that the truncation method requires us to start with the classical computation of the DPSS taper in the time domain. In the figure below, we illustrate the time domain values of the DPSS taper.
\begin{center}
\begin{tikzpicture}[scale=0.35]
\draw[->] (0,0) -- (32,0);
\foreach \x in {0,4,8,12,16,20,24,28,31}
    \node[below] at (\x,0) {\tiny \x};
\foreach \x/\h in {8/0.3, 9/1.0, 10/1.5, 11/2.2, 12/2.8, 13/3.5, 14/4.2, 15/4.8, 16/5.2, 17/4.8, 18/4.2, 19/3.5, 20/2.8, 21/2.2, 22/1.5, 23/1.0, 24/0.3}
    \draw[line width=1pt] (\x,0) -- (\x,\h);
\end{tikzpicture}
\end{center}
Next, we transform the time domain values to the frequency domain:
\begin{center}
\begin{tikzpicture}[scale=0.35]
\draw[->] (0,0) -- (32,0);
\foreach \x in {0,4,8,12,16,20,24,28,31}
    \node[below] at (\x,0) {\tiny \x};
\draw[line width=2pt] (0,0) -- (0,6);
\draw[line width=1.5pt] (1,0) -- (1,4);
\draw[line width=1pt] (2,0) -- (2,2.5);
\draw[line width=0.8pt] (3,0) -- (3,1.5);
\draw[line width=0.5pt] (4,0) -- (4,0.8);
\draw[line width=0.2pt] (5,0) -- (5,0.3);
\draw[line width=1.5pt] (31,0) -- (31,4);
\draw[line width=1pt] (30,0) -- (30,2.5);
\draw[line width=0.8pt] (29,0) -- (29,1.5);
\draw[line width=0.5pt] (28,0) -- (28,0.8);
\draw[line width=0.2pt] (27,0) -- (27,0.3);
\end{tikzpicture}
\end{center}
Now, we proceed with the truncation step where we choose to only keep the largest $2k+1 = 2^{m-1} = 7$ values centered around 0. These are the frequencies that will be associated to the phase values that we hope to measure with high probability.

\usetikzlibrary{decorations.pathreplacing}

\begin{center}
\begin{tikzpicture}[scale=0.35]
\draw[->] (0,0) -- (32,0);
\foreach \x in {0,4,8,12,16,20,24,28,31}
    \node[below] at (\x,0) {\tiny \x};
\draw[line width=2pt] (0,0) -- (0,6);
\draw[line width=1.5pt] (1,0) -- (1,4);
\draw[line width=1pt] (2,0) -- (2,2.5);
\draw[line width=0.8pt] (3,0) -- (3,1.5);
\draw[line width=1.5pt] (31,0) -- (31,4);
\draw[line width=1pt] (30,0) -- (30,2.5);
\draw[line width=0.8pt] (29,0) -- (29,1.5);
\draw[decorate,decoration={brace,amplitude=10pt,raise=2pt}] (4,2.5) -- (28,2.5);
\node[above] at (16,4) {set to zero};
\draw[dashed] (4,2.5) -- (4,0);
\draw[dashed] (28,2.5) -- (28,0);
\end{tikzpicture}
\end{center}
Once we have truncated away the unwanted frequencies, we proceed to the part where the task would be performed on the quantum computer. We begin with $m$ qubits and prepare the following state by using standard approaches for state synthesis:
\begin{center}
\begin{tikzpicture}[scale=0.3]
\draw[->] (0,0) -- (8,0);
\foreach \x in {0,4,7}
    \node[below] at (\x,0) {\tiny \x};
\draw[line width=2pt] (0,0) -- (0,6);
\draw[line width=1.5pt] (1,0) -- (1,4);
\draw[line width=1pt] (2,0) -- (2,2.5);
\draw[line width=0.8pt] (3,0) -- (3,1.5);
\draw[line width=1.5pt] (7,0) -- (7,4);
\draw[line width=1pt] (6,0) -- (6,2.5);
\draw[line width=0.8pt] (5,0) -- (5,1.5);
\node[right] at (11,2.5) {};
\node[below] at (9,-1.5) {};
\end{tikzpicture}
\end{center}
Note that we are only using $3$ qubits to encode the amplitudes of the 7 frequencies that we chose to keep. We now add 2 qubits in the $\ket{0}^{\otimes 2}$ state while adhering with the convention that the added qubits correspond to the two most significant bits in the computational state basis. The new state is then given by
\begin{center}
\begin{tikzpicture}[scale=0.3]
\draw[->] (0,0) -- (32,0);
\foreach \x in {0,4,8,12,16,20,24,28,31}
    \node[below] at (\x,0) {\tiny \x};
\draw[line width=2pt] (0,0) -- (0,6);
\draw[line width=1.5pt] (1,0) -- (1,4);
\draw[line width=1pt] (2,0) -- (2,2.5);
\draw[line width=0.8pt] (3,0) -- (3,1.5);
\draw[line width=1.5pt] (7,0) -- (7,4);
\draw[line width=1pt] (6,0) -- (6,2.5);
\draw[line width=0.8pt] (5,0) -- (5,1.5);
\node[right] at (33,2.5) {};
\end{tikzpicture}
\end{center}
Next, we apply the following unitary that comprises of several CNOTs:
\begin{align}
\prod_{n=m+1}^p \text{CNOT}_{m,n} = \text{CNOT}_{3,5}, \text{CNOT}_{3,4}
\end{align}
\vspace{0.5cm}
The effect of this set of gates on computational basis states is as follows:
\begin{align}
\begin{split}
\ket{0} &= \ket{00000} \quad \rightarrow \ket{11100} = \ket{0} \\
\ket{1} &= \ket{00001} \quad \rightarrow \ket{11100} = \ket{1}\\
\ket{2} &= \ket{00010} \quad \rightarrow \ket{11100} = \ket{2}\\
\ket{3} &= \ket{00011} \quad \rightarrow \ket{11100} = \ket{3}\\
\ket{4} &= \ket{00100} \quad \rightarrow \ket{11100} = \ket{28} \\
\ket{5} &= \ket{00101} \quad \rightarrow \ket{11101} = \ket{29} \\
\ket{6} &= \ket{00110} \quad \rightarrow \ket{11110} = \ket{30} \\
\ket{7} &= \ket{00111} \quad \rightarrow \ket{11111} = \ket{31} 
\end{split}
\end{align}
Note that the first four states remained unchanged under the CNOT gates. This is essential because we want to keep those four amplitudes on the left.
In general, this unitary shifts the right half of the taper by an amount $2^p - 2^m$, thus leaving it all the way to the right (that is, $\ket{2^{p} - 1}$). In other words, we have reproduced the original truncated frequency profile on a quantum register with $p$ qubits.
\begin{center}
\begin{tikzpicture}[scale=0.35]
\draw[->] (0,0) -- (32,0);
\foreach \x in {0,4,8,12,16,20,24,28,31}
    \node[below] at (\x,0) {\tiny \x};
\draw[line width=2pt] (0,0) -- (0,6);
\draw[line width=1.5pt] (1,0) -- (1,4);
\draw[line width=1pt] (2,0) -- (2,2.5);
\draw[line width=0.8pt] (3,0) -- (3,1.5);
\draw[line width=1.5pt] (31,0) -- (31,4);
\draw[line width=1pt] (30,0) -- (30,2.5);
\draw[line width=0.8pt] (29,0) -- (29,1.5);
\end{tikzpicture}
\end{center}
Finally, we apply QFT to our $p$-qubit register to transform back into the time-domain to get the following:
\begin{center}
\begin{tikzpicture}[scale=0.35]
\draw[->] (0,0) -- (32,0);
\foreach \x in {0,4,8,12,16,20,24,28,31}
    \node[below] at (\x,0) {\tiny \x};
\foreach \x/\h in {8/0.3, 9/0.85, 10/1.5, 11/2.3, 12/2.8, 13/3.7, 14/4.2, 15/5.1, 16/5.2, 17/4.8, 18/4.3, 19/3.5, 20/2.6, 21/2.2, 22/1.65, 23/1.0, 24/0.4}
    \draw[line width=1pt] (\x,0) -- (\x,\h);
\end{tikzpicture}
\end{center}
Note that the taper we have obtained is extremely close to the DPSS taper as proven in App.~\ref{app:TaperPreparation}

\section{Ideal Case Optimal Tapers} \label{app:QPDPSSDerivation}
The optimal taper for tQPE is the one that maximizes the probability of outputting the value of a phase estimate of the form $k/N$ that is $\delta$-close to $\theta$. To simplify the analysis we have set $\delta=2^{-(l+1)}$ for some integer $\ell$. In this case the $2K+1$ discrete frequencies closest to $\theta$ are indeed $\delta$-close to the true phase for $K=2^{m-1}$, where $p=\ell+m$ is the total number of ancilla qubits. At most a single discrete frequency that is $\delta$-close is left out, and that frequency has the smallest probability of being output by the algorithm. Thus without introducing much error we formulate the optimization problem in terms of maximizing the probability of outputting the closest $2K+1$ discrete estimates, see ~\eqref{eq:org-opt-prob}.
Since this is a constrained problem, we used the Lagrangian formulation to solve it. For ease of reference, we restate the corresponding Lagrangian~\eqref{eq:mainLagrange} from the main text below:
\begin{align}
     \mathcal{L}\left(|\phi\rangle, \lambda\right)
    = \sum_{j = -K}^{K}\left|\hat{\phi}\left(\Delta + \frac{j}{N}\right)\right|^2
    - \lambda\left(\sum_{n=0}^{N-1}\left|\phi[n]\right|^2 - 1\right).
\end{align}
To find the optimal taper for the original optimization problem, given by \eqref{eq:org-opt-prob}, we need to find the stationary point of $\mathcal{L}$ that maximizes the objective function of \eqref{eq:org-opt-prob}.
But before doing so, we expand the expression on the right-hand side of the above equation by plugging the definition of $\hat{\phi}$, given by \eqref{eq:FTofTaper}:
\begin{align}
     \mathcal{L}\left(\phi, \lambda\right) &=  \frac{1}{N}\sum_{j = -K}^{K} \sum_{n, n' = 0}^{N-1}\phi[n]\phi^*[n']e^{2\pi i \left(\Delta + \frac{j}{N}\right)(n-n')} -\lambda\left(\sum_{n=0}^{N-1}\phi[n]\phi^*[n] - 1\right) \\
    &= \frac{1}{N}\sum_{j = -K}^{K} \sum_{n, n' = 0}^{N-1}e^{2\pi i \Delta(n-n')}\phi[n]\phi^*[n']e^{-2\pi i (j/N)(n'-n)} -\lambda\left(\sum_{n=0}^{N-1}\phi[n]\phi^*[n] - 1\right)\label{eq:qp-dpss-objective}.
\end{align}
Then, to find the stationary points of $\mathcal{L}$, we set all partial derivatives to zero. First, we differentiate $\mathcal{L}$ with respect to $\phi^*[m]$ for all $m\in \{0, \ldots, N-1\}$ then set the expression to zero:
\begin{align}
    &\frac{\partial \mathcal{L}}{\partial \phi^*[m]} = \frac{1}{N}\sum_{j = -K}^{K} \sum_{n=0}^{N-1}e^{2\pi i \Delta(m-n)}\phi[n]e^{-2\pi i (j/N)(m-n)} -\lambda\phi[m] = 0 \\
    &\implies \frac{1}{N}\sum_{j = -K}^{K} \sum_{n=0}^{N-1} e^{2\pi i \Delta(n-m)}\phi[n]e^{-2\pi i (j/N)(m-n)} = \lambda\phi[m]\\
    & \implies \frac{1}{N}\sum_{n=0}^{N-1} e^{2\pi i \Delta(n-m)}\left (\sum_{j = -K}^K e^{-2\pi i (j/N)(m-n)}\right) \phi[n] = \lambda\phi[m] \\
    &\implies \frac{1}{N}\sum_{n=0}^{N-1}  e^{2\pi i \Delta(n-m)}\left(\frac{\sin\left(\pi (m-n)(2K+1)/N\right)}{\sin\left(\pi (m-n)/N\right)}\right)  \phi[n] = \lambda\phi[m]
    \label{eqn:odd_eigenvalue}
\end{align}
As is evident, the last equation is an eigenvalue equation with eigenvalue $\lambda$ and eigenvector $|\phi\rangle$. We rewrite this equation more concisely as follows:
\begin{align}
\mathcal{B}_{\Delta} |\phi\rangle = \lambda |\phi\rangle,
\end{align}
where we define the matrix $\mathcal{B}_{\Delta}$ as
\begin{align}
    \mathcal{B}_{\Delta}[m, n] \coloneqq \frac{1}{N} e^{2\pi i \Delta(n-m)}\left(\frac{\sin\left(\pi (m-n)(2K+1)/N\right)}{\sin\left(\pi (m-n)/N\right)}\right). \label{eq:Deltakernel}
\end{align}

Next, we differentiate $\mathcal{L}$ with respect to the Lagrange multiplier, $\lambda$, and also set it to zero:
\begin{align}
    & \frac{\partial \mathcal{L}}{\partial \lambda} = \sum_{n=0}^{N-1}\phi[n]\phi^*[n] - 1 = 0\\
    &\implies \sum_{n=0}^{N-1}\phi[n]\phi^*[n] = 1,\label{eq:stat_2_cond}
\end{align}
reproducing the normalization constraint.

Now, by plugging all the stationary points $(|\phi\rangle, \lambda)$ that satisfy the conditions \eqref{eqn:odd_eigenvalue} and \eqref{eq:stat_2_cond} into the objective function of the original optimization problem \eqref{eq:org-opt-prob}, we find:
\begin{align}
    \sum_{j = -K}^{K}\left|\hat{\phi}\left(\Delta + \frac{j}{N}\right)\right|^2 = \lambda.
\end{align}
In other words, this means that the stationary point  $(|\phi\rangle, \lambda)$ that maximizes the objective function is the eigenvector of $B_{\Delta}$ with the maximum eigenvalue $\lambda$, and this value corresponds to the success probability of the algorithm.

\section{Relationship Between P-DPSS Tapers and Ideal Case Optimal Tapers}\label{app:relationship}

We compare the operator $B_\Delta$ defined in~\eqref{eq:Deltakernel} to the bandlimiting operator $\mathcal{B}_K\,:\,\mathbb{C}^N \to \mathbb{C}^N$ described in Zhu, et al.~\cite{zhu2017eigenvalue}.
Suppose $K \in \mathbb{N}$ and $2K+1 < N$.
The bandlimiting operator, $\mathcal{B}_K$, is defined such that it zeros out the discrete frequencies of a signal $\phi \in \mathbb{C}^{N}$ that lie outside of the range $\{-K \equiv N - K \ (\textrm{mod}\ N), \ldots, 0, \ldots, K\}$. Formally, this operator does the following:
\begin{equation}
    \left(\mathcal{B}_K \left(\phi\right)\right)[m] \coloneqq 
        \frac{1}{\sqrt{N}}\sum_{n\in S}\hat{\phi}[n]e^{\frac{2\pi i m n}{N}}  \; ,
\end{equation}
where $m \in \left\{0, \ldots, N - 1\right\}$ and $S = \{0, 1, \ldots, K\} \cup \{N-K, \ldots, N-1\}$.
In the time domain this can be written as
\begin{equation}
    \left(\mathcal{B}_K \left(\phi\right)\right)[m] = \frac{1}{N} \sum_{n = 0}^{N-1}\frac{\sin\left(\pi (m-n)(2K+1)/N\right)}{\sin\left(\pi (m-n)/N\right)} \phi[n].
\end{equation}
Note that this expression only differs from the matrix description stated in \eqref{eqn:odd_eigenvalue} by a factor of $e^{2\pi i \Delta (m-n)}$ in each summand. We now introduce the timelimiting operator $\mathcal{T}_N\,:\,\mathbb{C}^N \to \mathbb{C}$ that is defined similarly as the following:
\begin{equation}
    \left(\mathcal{T}_N\left(\phi\right)\right)[n]\coloneqq \begin{cases}\phi[n],&\;\;\;n \in\{0, 1, \ldots, N-1\} \\
    0,&\;\;\;\text{else.}\end{cases}
\end{equation}
In~\cite{zhu2017eigenvalue}, it was shown that the P-DPSSs are the eigensequences of the operator $\mathcal{T}_N\mathcal{B}_K\mathcal{T}_N$.
In other words, the P-DPSSs are discrete sequences that contain a particular eigensequence that not only fulfills the time-limiting constraints of having non-zero values only on the indices $0, 1, \ldots, N-1$, but also has maximal frequency concentration within our window of interest, $2K+1$. 
Notice that the time-limiting property and the property of having maximal frequency concentration around $2K+1$ lattice points are the same properties that we want the optimal tapers to possess. 
These similarities make it interesting for us to study the eigenvalue spectrum of the P-DPSSs because it may suggest an upper bound on what we can do with our ideal tapers.
The bounds for the eigenvalues of the P-DPSS's with respect to the operator $\mathcal{T}_N\mathcal{B}_K\mathcal{T}_N$, proven in~\cite{zhu2017eigenvalue}
, are stated in the following theorem, which we restate with our notation.
\begin{thm}\label{th:Theorem1a} ~\cite[Theorem~1]{zhu2017eigenvalue}
    Suppose $K, N \in \mathbb{N}$ and $W = \frac{2K + 1}{2N} < \frac{1}{2}$. Also, suppose that $\lambda^{(i)}$ is the $i$\textsuperscript{th} largest eigenvalue of $\mathcal{T}_N\mathcal{B}_K\mathcal{T}_N$. Then for any $\epsilon \in \left(0, \frac{1}{2}\right)$, we have 
    \begin{equation}
        \lambda^{(2\lfloor NW \rfloor  - \lceil R(N, \epsilon)\rceil)} \geq 1 - \epsilon,
        \label{eq:eigenvalue-conc-karnick}
    \end{equation}
    where 
    \begin{align}\label{eqn:radius}
        R(N, \epsilon) &= \left(\frac{4}{\pi^2}\log(8N) + 6\right)\log\left(\frac{16}{\epsilon}\right) + 2\max\left(\frac{-\log\left(\frac{\pi}{32}\left(\left(\frac{N}{N-1}\right)^2 -1\right)\epsilon\right)}{\log\left(\frac{N}{N-1}\right)}, 0\right).
    \end{align}
\end{thm}

Although, Thm.~\ref{th:Theorem1a} is stated for the matrix that corresponds to $\mathcal{T}_N\mathcal{B}_K\mathcal{T}_N$, the authors of~\cite{zhu2017eigenvalue} showed that the eigenvalues of $\mathcal{T}_N\mathcal{B}_K\mathcal{T}_N$ and $\mathcal{B}_K$ are the same. 
Note that the matrix representation of the eigenvalue equation given by \eqref{eqn:odd_eigenvalue} differs from the matrix representation of $\mathcal{B}_K$ by some complex phase. 
However, we demonstrate in the following proposition that they share the same spectrum by showing that we can adapt the eigensequences of $\mathcal{B}_K$ to obtain the eigensequences that satisfy \eqref{eqn:odd_eigenvalue}. 
In other words, we can use the analytical bounds for the eigenvalues of the P-DPSSs to bound the eigenvalues for our ideal tapers.
Below, we refer to $\phi_i^{\Delta}$, the eigenvector of the operator described in \eqref{eqn:odd_eigenvalue} with eigenvalue $\lambda^{(i)}$, as the quantum periodic discrete prolate spheroidal sequence (QP-DPSS), where $-\frac{1}{2N} \leq \Delta \leq \frac{1}{2N}$.

\begin{propos}\label{propos:PDPSS-evalue}
Suppose $\{(\lambda^{(i)}, \psi_i)\}_i$ represents the set of eigenvalues and eigensequences of $\mathcal{B}_K$. Then, for some $-\frac{1}{2N} \leq \Delta \leq \frac{1}{2N}$,  $\{(\lambda^{(i)}, \phi_i^{\Delta})\}_i$ is the set of eigenvalues and eigensequences of $B_\Delta$ given in \eqref{eq:Deltakernel}, where
\begin{equation}\label{eqn:transformed_taper}
        \phi_i^{\Delta}[n] = e^{2\pi i \Delta n} \psi_i[n].
\end{equation}
    \end{propos}
    \begin{proof}
We first restate the eigenvalue equation described in \eqref{eqn:odd_eigenvalue} with the proposed eigensequence, $\phi_i^{\Delta}[n]$.
\begin{align}
    & \frac{1}{N}\sum_n  e^{2\pi i \Delta(m-n)}\left(\frac{\sin\left(\pi (m-n)(2K+1)/N\right)}{\sin\left(\pi (m-n)/N\right)}\right) \phi_i^{\Delta}[n]  \notag \\ 
    & = \frac{1}{N}\sum_n  e^{2\pi i \Delta(m-n)}\left(\frac{\sin\left(\pi (m-n)(2K+1)/N\right)}{\sin\left(\pi (m-n)/N\right)}\right) e^{2\pi i \Delta n}\psi_i[n] \\
    & = \frac{e^{2\pi i \Delta m}}{N}\sum_n  \left(\frac{\sin\left(\pi (m-n)(2K+1)/N\right)}{\sin\left(\pi (m-n)/N\right)}\right) \psi_i[n] \\
    &= e^{2\pi i \Delta m} \lambda^{(i)} \psi_i[m]\\
    &= \lambda^{(i)} \phi_i^{\Delta}[m]
\end{align}
where the second to last equality follows from the fact that $\psi_i$ is an eigensequence of $\mathcal{B}_K$ with eigenvalue $\lambda^{(i)}$. This concludes the proof.
\end{proof}

Note that the eigensequences $\{\phi_i^{\Delta}\}_{i}$ are what we refer to as QP-DPSSs in the main text.
It is clear from the development above that transforming the tophat taper in the way shown in \eqref{eqn:transformed_taper} would give us a taper that would succeed with probability 1.
In fact, it is known that the P-DPSSs have up to $2K+1$ eigensequences with eigenvalue 1.
While such a transformation would be amazing, we remark that it requires explicit knowledge of $\Delta$ which is not something that we can reasonably assume to possess (i.e. it would assume knowledge of the phase that we are trying to determine).
Nevertheless, we provide some further discussion below to allow the reader to understand the deeper connections between P-DPSSs and optimal quantum tapers.
In particular, we note that it may be possible to construct optimal tapers with greater efficiency by taking linear combinations of optimal tapers due to the possibility of the existence of a much more efficient description of the tapers.
In fact, if one is willing to give up on the goal of achieving a success probability of 1, there are possibly other tapers with a much more efficient representation.
From the development above in Thm.~\ref{th:Theorem1a} and Prop.~\ref{propos:PDPSS-evalue}, we have the following result.

\begin{corol}\label{co:m_qubits}
Given $N = 2^{l+m}$, where $l$ is the number of qubits needed to encode the target phase and $m$ is the number of extra qubits necessary to guarantee that with probability at least $1 -\epsilon$, the desired precision will be met. Using the QP-DPSS taper with the largest eigenvalue, we have
    \begin{equation}
        m = O\left(\log\log 1/\delta + \log\log 1/\epsilon\right).
    \end{equation}
\end{corol}
\begin{proof}
    We are interested in the largest eigenvalue, i.e., $\lambda^{(0)}$.
    Therefore, from \eqref{eq:eigenvalue-conc-karnick}, we first set
    $2\lfloor NW\rfloor - \lceil R(N, \epsilon) \rceil = 0$ to obtain $\lceil R(N, \epsilon) \rceil = 2\lfloor NW\rfloor \leq 2NW$. 
    For all $N \geq 2$, the first term in \eqref{eqn:radius} dominates and therefore, \eqref{eqn:radius} reduces to
    \begin{equation}
        R(N, \epsilon) = O\left(\log\left(\frac{1}{\epsilon}\right)\log N \right).
    \end{equation}
    This naturally implies that
    \begin{equation}
        NW = O\left(\log\left(\frac{1}{\epsilon}\right)\log N \right).
    \end{equation}
    Recall that $m$, the additional qubits required for a given precision, can be expressed as $\log\left(NW\right)$ because $W = \frac{2^m}{2^{l+m}}$ is the fraction of the lattice for which the taper has significant support. 
    This gives 
    \begin{align}
        2^m &= O\left(\log\left(\frac{1}{\epsilon}\right)\log N \right)\\
        &= O\left((m+l)\log(1/\epsilon)\right) \\
    &= O\left(m\log1/\epsilon + \log 1/\delta \log1/\epsilon\right) \\
    \implies m &= O\left(\log\left(m\log1/\epsilon + \log 1/\delta \log1/\epsilon\right)\right) \\
    & =  O\left(\log\left(\log 1/\delta \log1/\epsilon\right) + \frac{m\log 1/\epsilon}{\log 1/\delta \log 1/\epsilon}\right) \\
    &= O\left(\log\left(\log 1/\delta \log1/\epsilon\right) + \frac{m }{\log 1/\delta }\right) \\
    \implies m - O\left(\frac{m}{\log 1/\delta}\right)  &= O\left(\log\left(\log 1/\delta \log1/\epsilon\right)\right) \\
    \implies m &= \frac{O\left(\log\left(\log 1/\delta \log1/\epsilon\right)\right)}{1 - O\left(\frac{1}{\log 1/\delta}\right)} \\
    &= O\left(\log\log 1/\delta + \log\log 1/\epsilon\right).
    \end{align}
    The second equality follows by substituting $N = 2^{l+m}$. The fifth equality follows from the fact that $\log (x+y) \leq \log (x) + y/x$ for all $x, y > 0$. This concludes the proof.
\end{proof}

\section{Average-Case Optimal Tapers}\label{app:averagecase}
For ease of reference, we restate the
Lagrangian corresponding to the average case from the main text below:
\begin{equation}  \overline{\mathcal{L}}\left(|\phi\rangle, \lambda\right)
    =  \mathbb{E}_{\Delta \sim \mathcal{D}}\left[\sum_{j = -K}^{K}\left|\hat{\phi}\left(\Delta + \frac{j}{N}\right)\right|^2\right]
    - \lambda\left(\sum_{n=0}^{N-1}\left|\phi[n]\right|^2 - 1\right).\label{eq:avg-case-mainLagrange-app}
\end{equation}

As mentioned before, we focus on the case where $\Delta$ is uniformly distributed over the interval $[-\sfrac{1}{2N}, \sfrac{1}{2N}]$, i.e., $\mathcal{D}$ is the uniform distribution over this interval. It is reasonable to make this assumption since it is unlikely that the phase values are influenced in a way such that they lie away or close to the lattice points in a systematic way. However, it is possible to generalize the subsequent analysis for other probability distributions, such as the Gaussian distribution.

Please note that the analysis below follows a similar approach to that outlined in App.~\ref{app:QPDPSSDerivation}. Now, consider the following:
\begin{align}  & \overline{\mathcal{L}}\left(|\phi\rangle, \lambda\right) \nonumber\\
     & =  \int_{-\frac{1}{2N}}^{\frac{1}{2N}}\dd\Delta N \sum_{j = -K}^{K}\left|\hat{\phi}\left(\Delta + \frac{j}{N}\right)\right|^2
    -\lambda\left(\sum_{n=0}^{N-1}\left|\phi[n]\right|^2 - 1\right)\\
    & = \frac{1}{N}\int_{-\frac{1}{2N}}^{\frac{1}{2N}}\dd\Delta N\sum_{j=-K}^{K}\sum_{n, n'=0}^{N-1}\phi[n]\phi^*[n']e^{2\pi i (\Delta + \frac{j}{N})(n - n')} -\lambda\left(\sum_{n=0}^{N-1} \phi[n]\phi^*[n] - 1\right)\\
    & = \sum_{n, n'= 0}^{N-1}\phi[n]\phi^*[n'] \int_{-\frac{1}{2N}}^{\frac{1}{2N}}\dd\Delta e^{2\pi i \Delta(n-n')} \sum_{j=-K}^K e^{-2\pi i (j/N)(n'-n)} -\lambda\left(\sum_{n=0}^{N-1} \phi[n]\phi^*[n] - 1\right)\\
    & =  \sum_{n, n'= 0}^{N-1}\phi[n]\phi^*[n'] \frac{\sin\left(2\pi \left(1/2N\right)(n' - n)\right)}{\pi(n'-n)} \frac{\sin\left(\pi (n'-n)(2K+1)/N\right)}{\sin\left(\pi(n'-n)/N\right)} -\lambda\left(\sum_{n=0}^{N-1} \phi[n]\phi^*[n] - 1\right)\\
    & =  \sum_{n, n'= 0}^{N-1}\phi[n]\phi^*[n'] \underbrace{ \frac{\sin\left(\pi (n'-n)(2K+1)/N\right)}{\pi(n'-n)}}_{\eqqcolon C(n, n')} -\lambda\left(\sum_{n=0}^{N-1} \phi[n]\phi^*[n] - 1\right)\label{eq:avg-lag}
\end{align}
Then, to find stationary points of $\overline{\mathcal{L}}$, we set all of its partial derivatives to zero. First, we differentiate $\overline{\mathcal{L}}$ with respect to $\phi^*[m]$ for all $m\in \{0, \ldots, N-1\}$ and set it to zero:
\begin{align}
    & \frac{\partial\overline{\mathcal{L}}}{\partial\phi^*[m]} = \sum_{n=0}^{N-1}\phi[n] C(n, m) - \lambda\phi[m] = 0\\
    & \implies \sum_{n=0}^{N-1}\phi[n] C(n, m) = \lambda\phi[m]\label{eq:odd_eigenvalue_avg}.
\end{align}
Note that the above equation is an eigenvalue equation, and its eigenvectors are DPSSs, previously studied in classical signal processing~\cite{slepian1978prolate}. Next, differentiating $\mathcal{L}$ with respect to the Lagrange multiplier $\lambda$ and setting it to zero leads to the normalization constraint for the taper:
\begin{align}
    & \frac{\partial \overline{\mathcal{L}}}{\partial \lambda} = -\sum_{n=0}^{N-1}\phi[n]\phi^*[n] + 1 = 0\\
    &\implies \sum_{n=0}^{N-1}\phi[n]\phi^*[n] = 1,\label{eq:stat_2_cond_avg}
\end{align}

Next, by plugging all the stationary points $(|\phi\rangle, \lambda)$ that satisfy the conditions \eqref{eq:odd_eigenvalue_avg} and \eqref{eq:stat_2_cond_avg} into the objective function of $\overline{\mathcal{L}}$, we get the following:
\begin{align}
    \mathbb{E}_{\Delta \sim \mathcal{D}}\left[\sum_{j = -K}^{K}\left|\hat{\phi}\left(\Delta + \frac{j}{N}\right)\right|^2\right] = \lambda.
\end{align}
In other words, this means that the stationary point  $(|\phi\rangle, \lambda)$ that maximizes the objective function is the eigenvector of $C$ with the maximum eigenvalue $\lambda$. This corresponds to the DPSS with the maximum eigenvalue, which additionally corresponds to the success probability of the algorithm.

\section{Classical Signal Analysis Derivation of DPSS as Optimal Tapers} \label{app:classicaldpss}

Following~\cite{slepian1978prolate,karnik2021improved}, the discrete time Fourier transform (DTFT), $\hat{x} \in L_2([-\frac{1}{2},\frac{1}{2}])$, of a discrete signal, $x \in l_2(\mathbb{Z})$, is given by
\begin{equation}
  \hat{x}(f) \coloneqq \sum_{n = -\infty}^\infty x[n] e^{-2\pi i fn}, \;\;\; f \in \left[-\frac{1}{2}, \frac{1}{2}\right] \; ,
\end{equation}
with the inverse transform given by
\begin{equation}
  x[n] = \int_{-\frac{1}{2}}^\frac{1}{2} \hat{x}(f) e^{2\pi i fn} df \; .
\end{equation}
From these definitions, we can see that $x, x' \in l_2(\mathbb{Z})$ satisfy the Parseval-Plancherel inequality, 
$\langle x, x' \rangle_{l_2(\mathbb{Z})} = \langle \hat{x}, \hat{x}' \rangle_{L_2([-1/2,1/2])}$. Here, we denote $\langle \cdot, \cdot \rangle_{\mathcal{A}}$ as the inner product defined on the vector space $\mathcal{A}$.
We say that $x \in l_2(\mathbb{Z})$ is time-limited to $n \in \{ 0, \dots, N-1 \}$ if $x[n] = 0$ for all $n \in \mathbb{Z}\backslash\{ 0, \dots, N-1 \}$. Furthermore, we say that $x \in l_2(\mathbb{Z})$ is band-limited to $|f| \leq W$ if $\hat{x}(f) = 0$ for $|f| > W$, where $W \in (0, 1/2)$.

The aim here is to find discrete functions $x[n]$ that are time-limited to $n \in \{ 0, \dots, N-1 \}$ and are maximally band-limited to the frequency band $|f| \leq W$. We can formulate this as an optimization problem in the following way:
\begin{align}
\max_{x \in l_2(\mathbb{Z})}\quad& \int_{-W}^W |\hat{x}(f)|^2 df \label{eq:opt1}\\
    \text{subject to} \quad &  \left\Vert x\right\Vert^2_{l_{2}(\mathbb{Z})} = 1,\\
    \quad & x[n] = 0 \text{ for all } n \in \mathbb{Z}\backslash\{ 0, \dots, N-1 \}.
\end{align}
We can simplify things conceptually by defining the following two operators: the time-limiting operator $\mathcal{T}_N$ as
\begin{equation}
  (\mathcal{T}_Nx)[n] \coloneqq \left\{ \begin{array}{c} x[n] \;\; \mathrm{if} \;\; n \in \{ 0, \dots, N-1 \} \\
  0 \;\; \mathrm{if} \;\; n \in \mathbb{Z}\backslash \{0, \dots, N-1\} \end{array} \right .
\end{equation}
and the band-limiting operator $\mathcal{B}_W$ as
\begin{equation}
  (\mathcal{B}_Wx)[n] \coloneqq \sum_{l=-\infty}^\infty \frac{\sin[2W(l-n)]}{\pi(l-n)} x[l] \;\; \mathrm{for} \;\; n \in \mathbb{Z} \; .
\end{equation}
Note that the DTFT of the band-limiting operator is  $\widehat{\mathcal{B}_W x}(f) = \hat{x}(f)$ for $|f| \leq W$ and $\widehat{\mathcal{B}_Wx}(f) = 0$ for $|f| > W$.

With these definitions, we see that the integral in \eqref{eq:opt1} can be represented as
\begin{equation}
  \int_{-W}^W |\hat{x}(f)|^2 df = \langle \hat{x}, \widehat{\mathcal{B}_W x} \rangle_{L_{2}\left([-\frac{1}{2},\frac{1}{2}]\right)} \; \label{eq:obj-opt1}.
\end{equation}
From the Parseval-Plancherel theorem, we have
\begin{equation}
  \langle \hat{x}, \widehat{\mathcal{B}_W x} \rangle_{L_{2}\left([-\frac{1}{2},\frac{1}{2}]\right)} = \langle x, \mathcal{B}_W x \rangle_{l_2(\mathbb{Z})} \; .
\end{equation}
And then
\begin{eqnarray}
  \langle x, \mathcal{B}_W x \rangle_{l_2(\mathbb{Z})} & = & \langle \mathcal{T}_{N} x, \mathcal{B}_W \mathcal{T}_{N} x \rangle_{l_2(\mathbb{Z})} \\
  & = & \langle x, \mathcal{T}_{N} \mathcal{B}_W \mathcal{T}_{N} x \rangle_{l_2(\mathbb{Z})} \; ,\label{eq:selfad}
\end{eqnarray}
with the last expression, \eqref{eq:selfad}, holding because $\mathcal{T}_{N}$ is self-adjoint.

From \eqref{eq:opt1}, \eqref{eq:obj-opt1}, and  \eqref{eq:selfad}, we can see that the eigenvector, $\phi_0$, corresponding to the maximum eigenvalue, $\lambda_0$, of the matrix given by
\begin{equation}
\left[ \mathcal{T}_N \mathcal{B}_W \mathcal{T}_N \right]_{l,n} = \frac{\sin(2\pi W (l-n))}{\pi (l - n)} ; \;\; l,n \in \{0, \dots, N-1\}
\end{equation}
solves the maximization problem. This classical signal processing result gives the same optimal tapers that we find above for the quantum average-case optimal taper.

\section{Relationship Between Tapers and the Fourier Convolution Theorem}\label{app:Convolution}

Here, we aim to show a conceptually useful way to understand tapering, and hence tQPE, originating in the classical signal processing community. We begin by deriving the Fourier convolution theorem in a form useful for our discussion below. Starting with two functions of time, $\psi$ and $x$, we have that their pointwise product (denoted by '$\cdot$')
\begin{eqnarray}
\psi(t) \cdot x(t) & = & \left(\int_{-\infty}^\infty\hat{\psi}(g) e^{-2\pi i gt} dg\right) \cdot\left(\int_{-\infty}^\infty\hat{x}(f) e^{-2\pi i ft} df\right) \\
& = & \int_{-\infty}^\infty \int_{-\infty}^\infty \hat{\psi}(g) \hat{x}(f) e^{-2\pi i (g+f)t} dg df\\
& = & 
\int_{-\infty}^\infty \underbrace{\left( \int_{-\infty}^\infty \hat{\psi}(g) \hat{x}(h-g) dg \right)}_{(\hat{\psi}\; * \;\hat{x})} e^{-2\pi i ht}dh \; ,
\end{eqnarray}
where $*$ denotes a convolution.

Switching notation and denoting the continuous Fourier transform as $\mathcal{F}$, we see from above
\begin{eqnarray}
  \psi(t) \cdot x(t) & = & \mathcal{F}^{-1}\{\mathcal{F}\{\psi\}*\mathcal{F}\{x\}\} \\
  & \rightarrow & 
  \mathcal{F}\{\psi(t)\cdot x(t)\} = \mathcal{F}\{\psi\}*\mathcal{F}\{x \} \; .
\end{eqnarray}
Note that here, to be consistent with the signs of the complex exponentials, we use the opposite convention for our transform than is usual.

For tQPE, the amplitudes of the state encoded on the ancilla register before the inverse QFT are $\phi[n] e^{2\pi i \theta n}$. Here, $\theta \in [0,2\pi]$, but the taper, $\phi$ is a time-limited sequence with discrete support at equally spaced times, $n \Delta t$. We can represent this sequence as a continuous distribution, $\phi(t) \cdot\Sh(t)$, in the limit as $n \rightarrow \infty$, and $\Delta t \rightarrow 0$,
where $\Sh(t) = \sum_{n=0}^{N-1} \delta(t - n \Delta t)$ is a Dirac comb.
In this limit, $ \phi[n] e^{2\pi i \theta n} \rightarrow \left(\left( \phi(t) \cdot \Sh(t)\right) \cdot e^{2\pi i \theta t}\right) = \left(\left( \phi(t) \sum_{n=0}^{N-1} \delta(t - n \Delta t)\right) \cdot e^{2\pi i \theta t}\right)= \left(\left( \sum_{n=0}^{N-1} \phi(t) \delta(t - n \Delta t)\right) \cdot e^{2\pi i \theta t}\right)$. To apply the Fourier convolution theorem, we take $x(t) = e^{2\pi i \theta t}$ giving $\mathcal{F}\{x\}(f) = \delta(\theta - f)$, and the discrete and finite in time distribution, $\psi(t) = \sum_{n=0}^{N-1} \phi(t) \delta(t - n \Delta t)$. This then gives us
\begin{eqnarray}
\mathcal{F}\{\psi\}(f) & = & \int_{-\infty}^\infty \sum_{n=0}^{N-1} \phi(t) \delta(t - n \Delta t) e^{-2\pi i tf} dt \\
& = & \sum_{n=0}^{N-1} \int_{-\infty}^\infty \phi(t)  e^{-2\pi i tf} \delta(t - n \Delta t) dt \\
& = & 
\sum_{n=0}^{N-1} \phi[n] e^{-2\pi i n \Delta t f} \; .
\end{eqnarray}

Rescaling $f \Delta t \rightarrow f$ and similarly for $\theta$, this gives us, 
\begin{align}
\mathcal{F}\{\psi(t) \cdot x(t)\}(f) &= \left\{\mathcal{F}\{\psi\}*\mathcal{F}\{x\} \right\}(f) \\
&=  \left\{\left( \sum_{n=0}^{N-1} \phi[n] e^{-2\pi i n f} \right) * \delta(\theta-f)\right\}(f) \\
&= \sum_{n=0}^{N-1} \phi[n] e^{-2\pi i n (\theta - f)} \\
&= \mathcal{F}\{\phi\}(\theta - f)
\end{align}

Switching back to our previous notation, for tQPE, we have
\begin{equation}
  \mathcal{F}\{\psi(t) \cdot x(t)\} = \hat\phi(\theta - f) \; ,
\end{equation}
the coefficient in (\ref{eq:taper-tQPE-final-state}). tQPE samples this with the QFT at discrete frequencies, $k$, as $\mathrm{QFT}^\dagger[\psi(t) \cdot x(t)] = \hat\phi_k(\theta - k/N)$.

\section{Proof of Theorem \ref{thm:uncomputing}}\label{app:uncomputing}

In this appendix we prove Theorem~\ref{thm:uncomputing} of Sec.~\ref{sec:uncomputing}. For convenience we reproduce the theorem statement here. For the definitions refer to Sec.~\ref{sec:uncomputing}.
\setcounter{thm}{3}
\begin{reptheorem}{4}
    The following holds:
    \begin{equation}
        \frac{1}{2}\left \Vert \mathcal{T}_{CA} - \operatorname{Tr}_B \circ \mathcal{Q}^{\dagger}_{BA} \circ \widetilde{\mathcal{T}}_{CB} \circ \mathcal{Q}_{BA} \circ \mathcal{A}_{CA} \right\Vert_{\diamond} \leq 2 L_T \delta + 4\epsilon.
    \end{equation}
    Here, $\operatorname{Tr}_B$ denotes the partial-trace channel, where the register $B$ is traced out.
\end{reptheorem}
\begin{proof}
    Consider that
    \begin{multline}
        \frac{1}{2}\left \Vert \mathcal{T}_{CA} - \operatorname{Tr}_B \circ \mathcal{Q}^{\dagger}_{BA} \circ \widetilde{\mathcal{T}}_{CB} \circ \mathcal{Q}_{BA}\circ \mathcal{A}_{CA} \right\Vert_{\diamond} \\
        = \sup_{\rho_{RCA}} \frac{1}{2}\left \Vert \left(\mathcal{I}_{R}\otimes \mathcal{T}_{CA}\right) (\rho_{RCA}) - \left(\mathcal{I}_{R}\otimes \operatorname{Tr}_B \circ \mathcal{Q}^{\dagger}_{BA} \circ \widetilde{\mathcal{T}}_{CB} \circ \mathcal{Q}_{BA}\circ \mathcal{A}_{CA}\right)(\rho_{RCA}) \right\Vert_{1}\label{eq:dia-to-one-norm},
    \end{multline}
where $R$ is a reference system of arbitrary dimension and $\mathcal{I}_R$ is the identity channel acting on $R$.
Expanding the first term of the right-hand side of the above equation, we get
\begin{align}
    \left(\mathcal{I}_R \otimes \mathcal{T}_{CA}\right)(\rho_{RCA}) & = (I_R \otimes T_{CA}) \rho_{RCA} (I_R \otimes T^{\dagger}_{AC})\\
    & = \left(I_R \otimes \sum_{j_1} T(\theta_{j_1})_{C} \otimes |\psi_{j_1}\rangle\langle\psi_{j_1}|_A\right) \rho_{RCA} \left(I_R \otimes \sum_{j_2} T^{\dagger}(\theta_{j_2})_{C} \otimes |\psi_{j_2}\rangle\langle\psi_{j_2}|_A\right)\\
    & = \left( \sum_{j_1}  I_R \otimes T(\theta_{j_1})_{C} \otimes |\psi_{j_1}\rangle\langle\psi_{j_1}|_A\right) \rho_{RCA} \left( \sum_{j_2}I_R \otimes T^{\dagger}(\theta_{j_2})_{C} \otimes |\psi_{j_2}\rangle\langle\psi_{j_2}|_A\right)\label{eq:first-term-exp-1}.
\end{align}
This implies that the Kraus operator, which we denote by $V_{RCA}$, of the above channel is
\begin{align}
    V_{RCA} = \sum_{j_1}  I_R \otimes T(\theta_{j_1})_{C} \otimes |\psi_{j_1}\rangle\langle\psi_{j_1}|_A.
\end{align}
Therefore, this channel can be rewritten as
\begin{equation}
    \left(\mathcal{I}_R \otimes \mathcal{T}_{CA}\right)(\rho_{RCA}) = V_{RCA} \rho_{RCA} V_{RCA}^{\dagger}\label{eq:tar-channel-K-form}.
\end{equation}

Similarly, we get the Kraus representation of the second channel in~\eqref{eq:dia-to-one-norm}. For this, consider the following:
\begin{align}
    & \mathcal{I}_R \otimes \operatorname{Tr}_B \circ \mathcal{Q}^{\dagger}_{BA} \circ \widetilde{\mathcal{T}}_{CB} \circ \mathcal{Q}_{BA} \circ  \mathcal{A}_{CA} (\rho_{RCA}) \notag\\
    & = I_R \otimes \operatorname{Tr}_B \circ \mathcal{Q}^{\dagger}_{BA} \circ \widetilde{\mathcal{T}}_{CB} \circ \mathcal{Q}_{BA} \circ  (\rho_{RCA} \otimes |0\rangle\langle0|_B)\\
    & = I_R \otimes \operatorname{Tr}_B \circ \mathcal{Q}^{\dagger}_{BA} \circ \widetilde{\mathcal{T}}_{CB}\left( Q_{BA}\left(\rho_{RCA} \otimes |0\rangle\langle0|_B\right)Q^{\dagger}_{AB}\right)\\
    & = I_R \otimes \operatorname{Tr}_B \circ \mathcal{Q}^{\dagger}_{BA} \circ \widetilde{\mathcal{T}}_{CB}\Bigg( \sum_{j_1, x_1, y_1} f_{j_1x_1y_1}  |x_1\rangle\langle y_1|_B \otimes |\psi_{j_1}\rangle\langle\psi_{j_1}|_A \left(\rho_{RCA} \otimes |0\rangle\langle0|_B \right)\sum_{j_2, x_2, y_2} f^*_{j_2x_2y_2}  |y_2\rangle\langle x_2|_B \otimes |\psi_{j_2}\rangle\langle\psi_{j_2}|_A\Bigg)\\
    & = I_R \otimes \operatorname{Tr}_B \circ \mathcal{Q}^{\dagger}_{BA} \circ \widetilde{\mathcal{T}}_{CB}\Bigg( \sum_{j_1, j_2 x_1, x_2} f_{j_1x_10}f^*_{j_2x_20} \langle\psi_{j_1}|\rho_{RCA}|\psi_{j_2}\rangle \otimes   |x_1\rangle\langle x_2|_B \otimes |\psi_{j_1}\rangle\langle\psi_{j_2}|_A\Bigg)\\
    & = I_R \otimes \operatorname{Tr}_B \circ \mathcal{Q}^{\dagger}_{BA} \Bigg( \sum_{j_1, j_2 x_1, x_2} f_{j_1x_10}f^*_{j_2x_20} T\bigl(g(x_1)\bigr)_{C}\langle\psi_{j_1}|\rho_{RCA}|\psi_{j_2}\rangle T^{\dagger}\bigl(g(x_2)\bigr)_{C} \otimes   |x_1\rangle\langle x_2|_B \otimes |\psi_{j_1}\rangle\langle\psi_{j_2}|_A\Bigg)\\
    & = I_R \otimes \operatorname{Tr}_B \Bigg( \sum_{j_1, j_2 x_1, x_2, y_3, y_4} f_{j_1x_10}f^*_{j_2x_20} f^*_{j_1x_1 y_3}f_{j_2x_2 y_4} T\bigl(g(x_1)\bigr)_{C}\langle\psi_{j_1}|\rho_{RCA}|\psi_{j_2}\rangle \nonumber \\
    & \qquad \qquad T^{\dagger}\bigl(g(x_2)\bigr)_{C} \otimes   |y_3\rangle\langle y_4|_B \otimes |\psi_{j_1}\rangle\langle\psi_{j_2}|_A\Bigg)\\
    & = I_R \otimes \Bigg( \sum_{j_1, j_2 x_1, x_2, y_3} f_{j_1x_10}f^*_{j_2x_20} f^*_{j_1x_1 y_3}f_{j_2x_2 y_3} T\bigl(g(x_1)\bigr)_{C}\langle\psi_{j_1}|\rho_{RCA}|\psi_{j_2}\rangle T^{\dagger}\bigl(g(x_2)\bigr)_{C} \otimes |\psi_{j_1}\rangle\langle\psi_{j_2}|_A\Bigg)\\
     & =  \sum_{j_1, j_2}\Bigg( \sum_{x_1, x_2, y_3} f_{j_1x_10}f^*_{j_2x_20} f^*_{j_1x_1 y_3}f_{j_2x_2 y_3} \left(I_R \otimes T\bigl(g(x_1)\bigr)_{C}\right)\langle\psi_{j_1}|\rho_{RCA}|\psi_{j_2}\rangle \left(I_R \otimes T^{\dagger}\bigl(g(x_2)\bigr)_{C}\right)  \otimes |\psi_{j_1}\rangle\langle\psi_{j_2}|_A\Bigg)\label{eq:second-term-exp}\\
     & = \sum_{y_3} \Bigg(\underbrace{\sum_{j_1, x_1} f_{j_1x_10} f^*_{j_1x_1 y_3} \left(I_R \otimes T\bigl(g(x_1)\bigr)_{C} \otimes |\psi_{j_1}\rangle \langle\psi_{j_1}|_A\right)}_{\eqqcolon K_{y_3}}\Bigg)\nonumber \\ 
     & \qquad \qquad \rho_{RCA} \Bigg(\underbrace{\sum_{j_2, x_2} f^*_{j_2x_20} f_{j_2x_2 y_3}\left(I_R \otimes T^{\dagger}\bigl(g(x_2)\bigr)_{C} \otimes |\psi_{j_2}\rangle \langle\psi_{j_2}|_A\right)}_{K^{\dagger}_{y_3}} \Bigg)\\
     & = \sum_{y_3} K_{y_3} \rho_{RCA} K^{\dagger}_{y_3}.\label{eq:algo-channel-K-form}
\end{align}

Plugging~\eqref{eq:tar-channel-K-form} and~\eqref{eq:algo-channel-K-form} into~\eqref{eq:dia-to-one-norm} and removing the system labels for simplicity, we get
\begin{align}
    \frac{1}{2}\left \Vert \mathcal{T}_{AC} - \operatorname{Tr}_B \circ \mathcal{Q}^{\dagger}_{BA} \circ \widetilde{\mathcal{T}}_{CB} \circ \mathcal{Q}_{BA}\circ \mathcal{A}_{CA} \right\Vert_{\diamond} = \sup_{\rho} \frac{1}{2} \left \Vert V \rho V^{\dagger} - \sum_{y_3} K_{y_3} \rho K^{\dagger}_{y_3} \right \Vert_1.
\end{align}
In the perfect QPE case, uncomputation is exact, resulting in register $B$ being perfectly reset to the all-zeros state after applying $Q^{\dagger}$ (see Figure~\ref{fig:uncomputing}). Consequently, in this scenario, there exists only a single Kraus operator (corresponding to $y_3=0$), and this Kraus operator is the unitary operator $V$. For the imperfect QPE case, we expect this operator to remain approximately unitary, as we will demonstrate next. 

Continuing,
\begin{align}
    \frac{1}{2}\left \Vert \mathcal{T}_{AC} - \operatorname{Tr}_B \circ \mathcal{Q}^{\dagger}_{BA} \circ \widetilde{\mathcal{T}}_{CB} \circ \mathcal{Q}_{BA}\circ \mathcal{A}_{CA} \right\Vert_{\diamond} & = \sup_{\rho} \frac{1}{2}\left \Vert V \rho V^{\dagger} - K_0 \rho K_0^{\dagger} - \sum_{y_3 \neq 0} K_{y_3} \rho K^{\dagger}_{y_3} \right \Vert_1\\
    & \leq \sup_{\rho} \frac{1}{2}\left(\left \Vert V \rho V^{\dagger} - K_0 \rho K_0^{\dagger}  \right \Vert_1 +  \left \Vert \sum_{y_3 \neq 0} K_{y_3} \rho K^{\dagger}_{y_3} \right \Vert_1\right)\\
    & = \sup_{\rho} \frac{1}{2}\left \Vert V \rho V^{\dagger} - K_0 \rho K_0^{\dagger}  \right \Vert_1 + \sup_{\rho}  \frac{1}{2} \left \Vert \sum_{y_3 \neq 0} K_{y_3} \rho K^{\dagger}_{y_3} \right \Vert_1\label{eq:V-K0-Krest},
\end{align}
where the first inequality follows from the triangle inequality.

Let us begin by bounding the first term of the above equation:
\begin{align}
    \sup_{\rho} \left \Vert V \rho V^{\dagger} - K_0 \rho K_0^{\dagger}  \right \Vert_1 & = \sup_{\rho}\left\Vert V \rho V^{\dagger} - K_0\rho V^{\dagger} + K_0 \rho V^{\dagger} - K_0 \rho K_0^{\dagger} \right\Vert_1\\
    & \leq \sup_{\rho}\left\Vert V \rho V^{\dagger} - K_0 \rho V^{\dagger}\right\Vert_1 + \sup_{\rho}\left\Vert K_0 \rho V^{\dagger} - K_0 \rho K_0^{\dagger} \right\Vert_1\\
    & = \sup_{\rho}\left\Vert (V-K_0)\rho V^{\dagger}\right\Vert_1 + \sup_{\rho}\left\Vert K_0 \rho (V^{\dagger} -K_0^{\dagger})\right\Vert_1\\
    & \leq \sup_{\rho}\left\Vert V-K_0\right\Vert \left\Vert\rho\right\Vert_1 \left \Vert V^{\dagger}\right\Vert + \sup_{\rho}\left\Vert K_0\right\Vert \left\Vert\rho\right\Vert_1 \left \Vert V^{\dagger} -K_0^{\dagger}\right\Vert\\
    & = \left\Vert V-K_0\right\Vert + \left\Vert K_0\right\Vert \left \Vert V-K_0\right\Vert\\
    & = (1+\left\Vert K_0\right\Vert)\left\Vert V-K_0\right\Vert\label{eq:U-K-bound},
\end{align}
where the first inequality follows from the triangle inequality and the second inequality follows from the Hölder's inequality. Next, we bound the quantities $\left\Vert K_0\right\Vert$ and $\left\Vert V-K_0\right\Vert$ from above one by one:
\begin{align}
    \left\Vert K_0\right\Vert & = \left \Vert \sum_{j_1, x_1} f_{j_1x_10} f^*_{j_1x_10} \left(I_R \otimes T(g(x_1))_{C} \otimes |\psi_{j_1}\rangle \langle\psi_{j_1}|_A\right) \right\Vert\\
    & = \max_{j_1} \left\Vert \sum_{x_1} |f_{j_1x_10}|^2  \left(I_R \otimes T(g(x_1))_{C}\right) \right\Vert\\
    & \leq \max_{j_1} \sum_{x_1} |f_{j_1x_10}|^2  \left\Vert \left(I_R \otimes T(g(x_1))_{C}\right) \right\Vert\\
    & \leq 1\label{eq:bounding-K0}.
\end{align}
The second equality follows from Lemma~\ref{lem:direct-sum-norm}, the first inequality follows from the triangle inequality, and the last inequality follows from the fact that for every $x_1$, we have $|f_{j_1x_10}|^2 \leq 1$. We now bound $\left\Vert V-K_0\right\Vert$ from above:
\begin{align}
    & \left\Vert V-K_0\right\Vert\notag\\
    & = \left\Vert \sum_{j_1}I_R \otimes T(\theta_{j_1})_{C} \otimes |\psi_{j_1}\rangle\langle\psi_{j_1}|_A-\sum_{j_1, x_1} f_{j_1x_10} f^*_{j_1x_10} \left(I_R \otimes T(g(x_1))_{C} \otimes |\psi_{j_1}\rangle \langle\psi_{j_1}|_A\right)\right\Vert\\
    & = \left\Vert \sum_{j_1}T(\theta_{j_1})_{C} \otimes |\psi_{j_1}\rangle\langle\psi_{j_1}|_A-\sum_{j_1, x_1} |f_{j_1x_10}|^2 \left(T(g(x_1))_{C} \otimes |\psi_{j_1}\rangle \langle\psi_{j_1}|_A\right)\right\Vert\\
    & = \max_{j_1}\left\Vert T(\theta_{j_1})_{C}-\sum_{x_1} |f_{j_1x_10}|^2 T(g(x_1))_{C}\right\Vert \label{eq:step-error}\\
    &= \max_{j_1}\left\Vert \sum_{x_1} |f_{j_1x_10}|^2 T(\theta_{j_1})_{C}-\sum_{x_1} |f_{j_1x_10}|^2 T(g(x_1))_{C}\right\Vert\\
    & \leq \max_{j_1}\sum_{x_1} |f_{j_1x_10}|^2 \left \Vert T(\theta_{j_1})_{C}-T(g(x_1))_{C}\right\Vert\\
    & = \max_{j_1}\left(\sum_{x_1:|\theta_{j_1}-g(x_1)|\leq \delta} |f_{j_1x_10}|^2 \left \Vert T(\theta_{j_1})_{C}-T(g(x_1))_{C}\right\Vert + \sum_{x_1:|\theta_{j_1}-g(x_1)|> \delta} |f_{j_1x_10}|^2 \left \Vert T(\theta_{j_1})_{C}-T(g(x_1))_{C}\right\Vert \right)\\
    & \leq  \max_{j_1}\left(\sum_{x_1:|\theta_{j_1}-g(x_1)|\leq \delta} |f_{j_1x_10}|^2 (L_T \delta) + \sum_{x_1:|\theta_{j_1}-g(x_1)|> \delta} |f_{j_1x_10}|^2 (2) \right)\\
    & \leq  \max_{j_1}\left( L_T \delta + 2\epsilon \right)\\
    & = L_T \delta + 2\epsilon,\label{eq:bounding-mvt} 
\end{align}
where the first inequality follows from the multiplicativity of the spectral norm under tensor products, the second equality follows from Lemma~\ref{lem:direct-sum-norm}, the second inequality follows from the triangle inequality, the third inequality follows from~\eqref{eq:lipschitz}, and the last inequality follows from the following fact about the error probability of the QPE algorithm used:
\begin{equation}
    \sum_{x_1:|\theta_{j_1}-g(x_1)|> \delta} |f_{j_1x_10}|^2 \leq \epsilon.
\end{equation}

To this end, using the bounds given by~\eqref{eq:bounding-K0} and~\eqref{eq:bounding-mvt} in~\eqref{eq:U-K-bound}, we get
\begin{align}
    \sup_{\rho} \left \Vert V \rho V^{\dagger} - K_0 \rho K_0^{\dagger}  \right \Vert_1 \leq 2 L_T \delta + 4\epsilon\label{eq:first-term-bounding}.
\end{align}

Now that we have bounded the first term of~\eqref{eq:V-K0-Krest}, we focus on bounding its second term. For this, consider the following:
\begin{align}
\left \Vert \sum_{y_3 \neq 0} K_{y_3}^{\dagger}K_{y_3} \right \Vert 
    & = \left \Vert I - K_0^{\dagger} K_0  \right \Vert\\
    & = \left \Vert V^{\dagger}V - K_0^{\dagger} K_0  \right \Vert\\
    & = \left \Vert V^{\dagger}V -  K_0^{\dagger} V + K_0^{\dagger} V -  K_0^{\dagger} K_0   \right \Vert\\
    & \leq \left \Vert V^{\dagger}V -  K_0^{\dagger} V \right \Vert + \left \Vert K_0^{\dagger} V -  K_0^{\dagger} K_0 \right \Vert\\
    & \leq \left \Vert V^{\dagger} - K_0^{\dagger}   \right \Vert \left \Vert V \right \Vert + \left \Vert K_0^{\dagger} \right \Vert \left \Vert V - K_0  \right \Vert \\
    & \leq 2 \left \Vert V- K_0  \right \Vert\\
    & \leq 2 L_T \delta + 4\epsilon\label{eq:bounding-Krest},
\end{align}
where the first inequality follows from the triangle inequality, the second inequality follows from the submultiplicativity property of the spectral norm, the third inequality follows from~\eqref{eq:bounding-K0}, and the last inequality follows from~\eqref{eq:bounding-mvt}. Next, we have the following equality:
\begin{align}
    \sup_{\rho} \left\Vert\sum_{y_3 \neq 0 } K_{y_3} \rho K^{\dagger}_{y_3}\right\Vert_1 = \left\Vert\sum_{y_3 \neq 0 }  K^{\dagger}_{y_3}K_{y_3}\right\Vert\label{eq:y_3_decompose}.
\end{align}
This follows due to the following fact, where each $A_i$ is some linear operator:
\begin{align}
    \left\Vert\sum_{i}  A^{\dagger}_{i}A_{i}\right\Vert & = \sup_{\rho} \operatorname{Tr}\left[\sum_{i}  A^{\dagger}_{i}A_{i} \rho\right]\label{eq:op-trace-first}\\
    & = \sup_{\rho} \sum_{i} \operatorname{Tr}\left[  A^{\dagger}_{i}A_{i} \rho\right]\\
    & = \sup_{\rho} \sum_{i} \operatorname{Tr}\left[  A_{i} \rho A^{\dagger}_{i}\right]\\
    & = \sup_{\rho}  \operatorname{Tr}\left[  \sum_{i}A_{i} \rho A^{\dagger}_{i}\right]\\
    & = \sup_{\rho}  \left\Vert \sum_{i}A_{i} \rho A^{\dagger}_{i}\right\Vert_1.\label{eq:op-trace-last}
\end{align}
Using the identity given by~\eqref{eq:y_3_decompose}, and the bound given in~\eqref{eq:bounding-Krest}, we get
\begin{align}
    \sup_{\rho} \left \Vert \sum_{y_3 \neq 0} K_{y_3} \rho K_{y_3}^{\dagger} \right\Vert_1 \leq 2 L_T \delta + 4\epsilon\label{eq:y3-not-0-terms}. 
\end{align}

Finally, combining the bounds of~\eqref{eq:first-term-bounding} and~\eqref{eq:y3-not-0-terms} to bound~\eqref{eq:V-K0-Krest} from above, we get
\begin{equation}
    \frac{1}{2}\left \Vert \mathcal{T}_{AC} - \operatorname{Tr}_B \circ \mathcal{Q}^{\dagger}_{BA} \circ \widetilde{\mathcal{T}}_{CB} \circ \mathcal{Q}_{BA}\circ \mathcal{A}_{CA} \right\Vert_{\diamond} \leq 2 L_T \delta + 4\epsilon.
\end{equation}
This concludes the proof.
\end{proof}

\section{Useful Lemmas}
\begin{lem}[Norm of direct-sum of operators]\label{lem:direct-sum-norm}
    Let $A = A_1 \bigoplus \cdots \bigoplus A_n$, where each $A_i$ is a linear operator and $n$ is some natural number. Then the following is true:
    \begin{equation}
        \left \Vert A \right \Vert = \max_i \left \Vert A_i \right \Vert.
    \end{equation}
\end{lem}
\begin{proof}
    For each $i \in \{1, \ldots, n\}$, consider the following singular value decomposition: $A_i = U_i D_i V_i^{\dagger}$, where $U_i$ and $V_i$ are unitaries and $D_i$ is a diagonal matrix consisting of singular values. Then the singular value decomposition of $A = UDV^{\dagger}$, where $U = U_1 \bigoplus \cdots \bigoplus U_n$, $D = D_1 \bigoplus \cdots \bigoplus D_n$, and $V = V_1 \bigoplus \cdots \bigoplus V_n$. Using this and the fact that the spectral norm of a matrix is equal to its maximum singular value, we get $\left \Vert A \right \Vert = \max_i \left \Vert A_i \right \Vert$. 
\end{proof}

\begin{lem}[Bound on diamond distance between unitary channels; Alternate simple proof of Lemma~12.6 of~\cite{10.1145/276698.276708}]\label{lem:dia-to-op}
    Let $\mathcal{U}(\cdot) \coloneqq U(\cdot)U^{\dagger}$ and $\mathcal{V}(\cdot) \coloneqq V(\cdot)V^{\dagger}$ be two unitary channels with unitaries $U$ and $V$, respectively. Then the following holds:
    \begin{equation}
        \frac{1}{2}\left \Vert \mathcal{U}-\mathcal{V}\right\Vert_{\diamond} \leq \left \Vert U-V \right\Vert.
    \end{equation}
\end{lem}
\begin{proof}
\begin{align}
    \frac{1}{2}\left \Vert \mathcal{U}-\mathcal{V}\right\Vert_{\diamond}
    & = \sup_{\rho} \frac{1}{2}\left \Vert (\mathcal{I} \otimes \mathcal{U})(\rho)- (\mathcal{I} \otimes \mathcal{V})(\rho)\right\Vert_1\\
    & = \sup_{\rho} \frac{1}{2}\left \Vert (I \otimes U)\rho(I \otimes U^{\dagger})- (I \otimes V)\rho (I \otimes V^{\dagger})\right\Vert_1\\
    & = \sup_{\rho} \frac{1}{2}\Big \Vert (I \otimes U)\rho(I \otimes U^{\dagger}) - (I \otimes U)\rho(I \otimes V^{\dagger}) + (I \otimes U)\rho(I \otimes V^{\dagger})- (I \otimes V)\rho (I \otimes V^{\dagger})\Big\Vert_1\\
    & \leq  \sup_{\rho} \frac{1}{2}\left \Vert (I \otimes U)\rho(I \otimes U^{\dagger}) - (I \otimes U)\rho(I \otimes V^{\dagger})\right\Vert_1 + \frac{1}{2}\left\Vert(I \otimes U)\rho(I \otimes V^{\dagger})- (I \otimes V)\rho (I \otimes V^{\dagger})\right\Vert_1\\
    & = \sup_{\rho} \frac{1}{2}\left \Vert (I \otimes U)\rho(I \otimes U^{\dagger} - I \otimes V^{\dagger})\right\Vert_1 + \frac{1}{2}\left\Vert(I \otimes U - I \otimes V)\rho(I \otimes V^{\dagger})\right\Vert_1\\
    & \leq  \sup_{\rho} \frac{1}{2}\left \Vert (I \otimes U) \right \Vert \left \Vert\rho \right \Vert_1 \left \Vert(I \otimes U^{\dagger} - I \otimes V^{\dagger})\right\Vert + \frac{1}{2}\left\Vert(I \otimes U - I \otimes V)\right\Vert\left\Vert\rho\right\Vert_1\left\Vert(I \otimes V^{\dagger})\right\Vert\\
    & = \frac{1}{2}  \left \Vert(I \otimes U^{\dagger} - I \otimes V^{\dagger})\right\Vert  + \frac{1}{2}\left\Vert(I \otimes U - I \otimes V)\right\Vert\\
    & = \frac{1}{2}  \left \Vert(I \otimes U - I \otimes V)\right\Vert  + \frac{1}{2}\left\Vert(I \otimes U - I \otimes V)\right\Vert\\
    & = \frac{1}{2}  \left \Vert U - V\right\Vert  + \frac{1}{2}\left\Vert U - V\right\Vert\\
    & = \left \Vert U - V\right\Vert.
\end{align}
The first inequality follows from the triangle inequality, the second inequality follows from the generalized Hölder's inequality, and the third inequality holds due to the multiplicativity property of the trace norm under tensor products.
\end{proof}

\end{document}